\documentclass{IEEEtran}
\usepackage{cite}
\usepackage{amsmath,amssymb,amsfonts}
\usepackage{graphicx}
\usepackage{textcomp}
\usepackage{xifthen}
\usepackage{arydshln}

\usepackage{algorithm}
\usepackage{moreverb}
\usepackage{algorithm}
\usepackage{algpseudocode}
\usepackage{scalefnt}
\usepackage{caption}
\usepackage{subcaption}
\usepackage{epsfig} 

\usepackage{amsthm}
\usepackage{mathrsfs}
\usepackage{cite}
\usepackage[usenames]{color}
\usepackage{multirow}
\usepackage{pbox}
\usepackage{accents}
\usepackage{ulem}
\usepackage[hidelinks]{hyperref}

\newcommand{\R}{\mathbb{R}}

\newcommand{\Forall}{\,\forall\,}

\newcommand{\m}[1]{\mathcal{#1}}
\newcommand{\MR}[1]{
	\ifthenelse{\isempty{#1}}{\mathcal{S}_{\infty}}{\mathcal{S}_{#1,\infty}}
}
\newcommand{\Reach}[3]{\mathrm{Reach}^{#1}_{#2}(#3)}

\def\BibTeX{{\rm B\kern-.05em{\sc i\kern-.025em b}\kern-.08em
		T\kern-.1667em\lower.7ex\hbox{E}\kern-.125emX}}

\theoremstyle{definition}
\newtheorem{remark}{\textbf{Remark}}
\newtheorem{definition}{\textbf{Definition}}
\newtheorem{assumption}{\textbf{Assumption}}
\newtheorem{proposition}{\textbf{Proposition}}
\newtheorem{lemma}{\textbf{Lemma}}
\newtheorem{problem}{\textbf{Problem}}
\newtheorem*{PinwheelProblem}{\textbf{Pinwheel Problem} (PP)}
\newtheorem*{WindowsSchedulingProblem}{\textbf{Windows Scheduling Problem} (WSP)}
\newtheorem{example}{Example}

\theoremstyle{theorem}
\newtheorem{theorem}{\textbf{Theorem}}

\begin{document}
	\title{Robust Control Invariance and Communication Scheduling in Lossy Wireless Networked Control Systems}
	\author{Masoud Bahraini, Mario Zanon, Paolo Falcone, and Alessandro Colombo
		\thanks{This work was partially supported by the Wallenberg Artificial Intelligence, Autonomous Systems and Software Program (WASP) funded by Knut and Alice Wallenberg Foundation. }
		\thanks{Masoud Bahraini is with Chalmers University of Technology, G\"oteborg, Sweden (e-mail: masoudb@chalmers.se). }
		\thanks{Mario Zanon is with IMT School for Advanced Studies Lucca, Italy (e-mail: mario.zanon@imtlucca.it).}
		\thanks{Paolo Falcone is with Chalmers University of Technology and DIEF, Universit\`a di Modena e Reggio Emilia, Italy (e-mail: paolo.falcone@unimore.it).}
		\thanks{Alessandro Colombo is with Politecnico di Milano, Italy (e-mail: alessandro.colombo@polimi.it).}}
	
	\maketitle
	\begin{abstract}
		In Networked Control Systems (NCS) impairments of the communication channel can be disruptive to stability and performance. In this paper, we consider the problem of scheduling the access to limited communication resources for a number of decoupled closed-loop systems subject to state and input constraint. The control objective is to preserve the invariance property of local state and input sets, such that constraint satisfaction can be guaranteed. 
			Offline and online, state feedback scheduling policies are proposed and illustrated through numerical examples, also in case the network is subject to packet losses.
	\end{abstract}
	
	\section{Introduction}
	\label{section:introduction}
	The problem of keeping the state of a dynamical system in a compact subset of its state space over an infinite time horizon in the presence of uncertainties was introduced more than four decades ago \cite{bertsekas1971minimax,Bertsekas72a}. Since then, \textit{reachability} analysis has been exploited in different applications, e.g., it is used in \textit{model checking} and \textit{safety verification} to ensure the system does not enter an unsafe region, some specific situations are avoided, and some properties for an acceptable design are met \cite{prajna2004safety,bouajjani1997reachability,gillula2010design}. Reachability analysis has several applications in \textit{model predictive control} such as \textit{terminal set} design, \textit{persistence feasibility} \cite{rawlings2008unreachable}, and robust invariant set computation for initial conditions \cite{gupta2019full}.
	
	Recently, progress in wireless communication technologies has provided new opportunities but also new challenges to control theorists. On the one hand, the use of communication in the control loop has several benefits, such as reduced system wiring, ease of maintenance and diagnosis, and ease of implementation \cite{zhang2001stability}. On the other hand, wireless links are subject to noise, time varying delay, packet loss, jitter, limited bandwidth, and quantization errors, which are not negligible in the stability and performance analysis. Feedback control systems that are closed through a network are called \textit{networked control systems} (NCS). See \cite{zhang2016survey} for recent advances in the field.
	
	An interesting subset of NCS problems, which has received little attention so far, regards a scenario where a central decision maker is in charge of keeping the states of a set of independent systems within given admissible sets, receiving measurement data and deciding which local controller(s) should receive state measurement(s) update through the common communication channel(s), see Fig.~\ref{fig:network}.
	\begin{figure}[t]
		\centering
		\includegraphics[width=1\columnwidth]{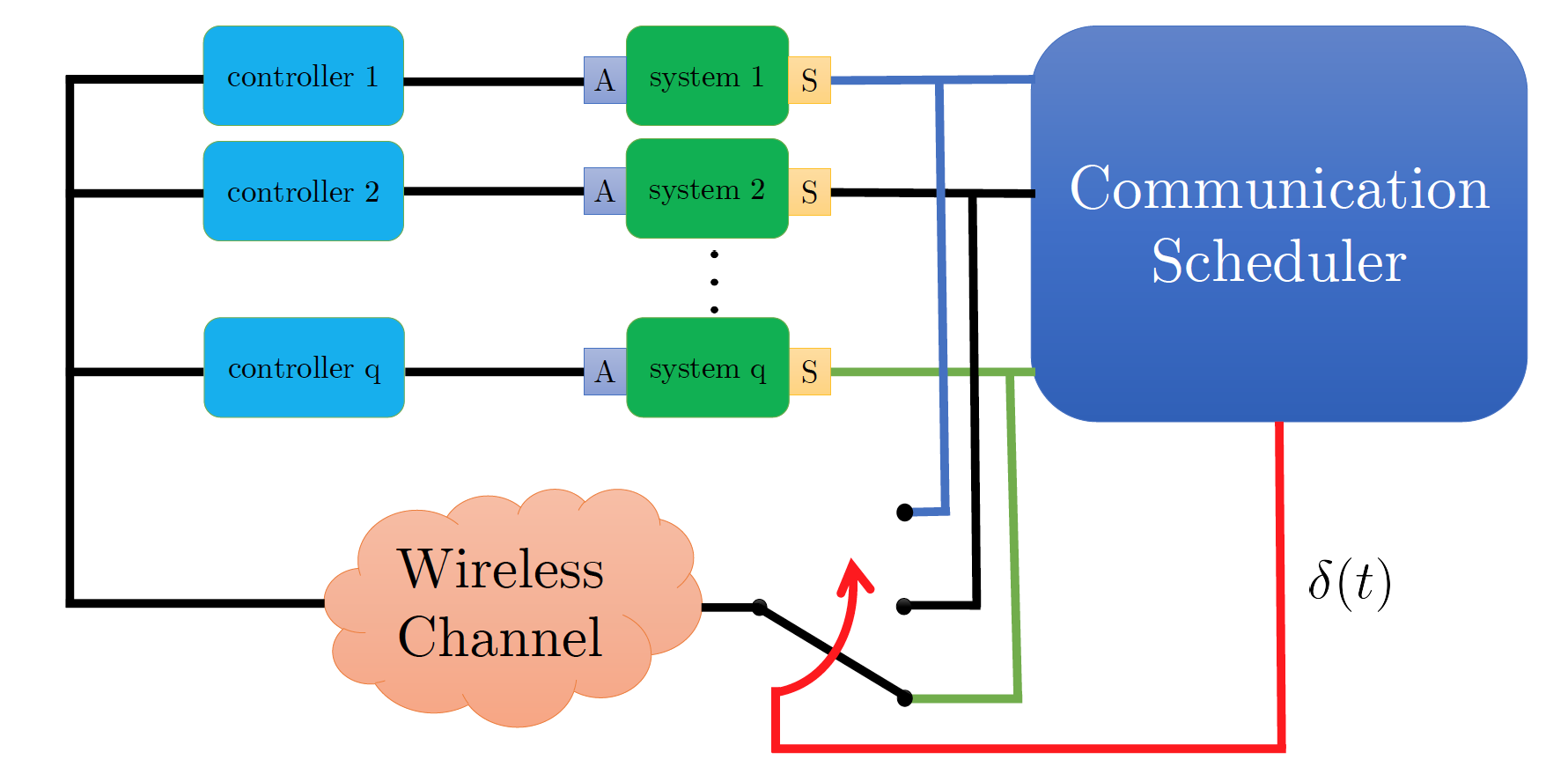}
		\caption{A set of agents sharing a common channel to exchange state measurements data based on decisions of the Central Scheduler.}
		\label{fig:network}
	\end{figure}
	This is a model, for instance, of remotely sensed and actuated robotic systems based on CAN communication or, as we will discuss in our application example, of remote multi-agent control setups for intelligent transportation systems field testing.  This scenario shares some similarities with event-driven control~\cite{Heemels12a}.  However, in our case, the core problem is to guarantee invariance of the admissible sets despite the communication constraints, rather than to ensure stability while minimizing communication costs, which does not explicitly guarantee satisfaction of hard constraints.  As we will see, this shift in focus brings about a corresponding shift in the set of available tools.
	
	We establish a connection between the control design and classical scheduling problems in order to use available results from the scheduling literature. These scheduling problems are the \textit{pinwheel problem} (PP) and the \textit{windows scheduling problem} (WSP)~\cite{Chan92a,Holte89a,bar2003windows}.  Both problems have been extensively studied and, though they are NP-hard, several heuristic algorithms have been proposed, which are able to solve a significant amount of problem instances~\cite{bar2002minimizing,bar2003windows,bar2007windows}.
	
	In this paper we target a reachability and safety verification problem, in discrete time, for the described NCS.  With respect to a standard reachability problem, the limitations of the communication channel imply that only a subset of the controllers can receive state measurements and/or transmit the control signals at any given time.  Therefore, a suitable \textit{scheduler} must be designed concurrently with the control law to guarantee performance. We formalize a general model for the control problem class, and propose a heuristic that solves the schedule design problem exploiting its analogy with  PP. Furthermore, we show that in some cases, our problem is equivalent to either  WSP or  PP. This gives us a powerful set of tools to co-design scheduling and control algorithms, and to provide guarantees on persistent schedulability. Building on these results, we propose online scheduling techniques to improve performance and cope with lossy communication channels.
	
	The rest of the paper is organized as follows. In Section~\ref{section:problem-statement}, the problem is formulated. Relevant mathematical background is stated in Section~\ref{section:math-background}. Our main contributions are divided into an offline and an online strategy, presented in Section~\ref{section:results_offline} and Section~\ref{section:results_online}, respectively. Examples and numerical simulations are provided in Section~\ref{section:numerics}.
	
	\section{Problem statement}
	\label{section:problem-statement}
	In this section, we define the control problem for uncertain constrained multi-agent NCS and provide the background knowledge needed to support Section~\ref{section:results_offline}. We first formulate the problem in general terms; afterwards, we provide some examples with different network topologies.
	\subsection{Problem Formulation}
	For each agent, $x$, $\hat{x}$ and $u$ denote the state of the plant, the state of the observer and the state of the controller, respectively. Note that in NCS the controller typically has memory to cope with packet losses. The discrete-time dynamics of the agent is described by
	\begin{equation}
	\label{eq:general_system}
	z^+=
	\begin{cases}
	F\big(z,v\big), &\textrm{if } \delta = 1,\\
	\hat{F}\big(z,v\big), &\textrm{if } \delta = 0,\\
	\end{cases}
	\end{equation}
	with
	\begin{equation}
	F:=\begin{pmatrix}
	f\big(z,v\big)\\
	g\big(z\big)\\
	c\big(z\big)
	\end{pmatrix}, ~
	\hat{F}:=\begin{pmatrix}
	\hat{f}\big(z,v\big)\\
	\hat{g}\big(z\big)\\
	\hat{c}\big(z \big)
	\end{pmatrix}, ~
	z:=\begin{pmatrix}
	x\\
	\hat{x}\\
	u
	\end{pmatrix},
	\end{equation}
	where $z \in\m{X}\times \m{X}\times \m{U}\subseteq\R^{n}\times\R^{n}\times\R^{r}$ and 
	$v \in \m{V} \subset \R^{p}$ denotes an external disturbance.   The two dynamics correspond to the \textit{connected mode}, i.e., $\delta=1$ and the \textit{disconnected mode}, i.e., $\delta=0$. The latter models the case in which the agent cannot communicate though the network and evolves in open loop.
	
	We consider $q$ agents of the form \eqref{eq:general_system}, possibly with different functions $F_i$, $\hat{F}_i$, $i\in\{1,\ldots,q\}$, and with state spaces of different sizes. We use the set $\boldsymbol{\mathcal{C}}=\{C_j\},~j \in \{1,\ldots,l\}$ of \textit{connection patterns} to represent the sets of agents that can be connected simultaneously.
	
	A connection pattern is an ordered tuple 
	\begin{equation}
	\label{eq:binary_vectors}
	C_j:=(c_{1,j},\ldots,c_{m_j,j}), 
	\end{equation}  
	with~$c_{k,j} \in \{1,\ldots,q\}$, $m_j \leq q$, and
	\begin{equation}
	i\in C_j \Leftrightarrow \delta_i=1 \textrm{ in connection pattern } j.
	\end{equation}
	For example, $C_1=(1,5)$ means that agents $1$ and $5$ are connected when $C_1$ is chosen.
	
	We can now formulate the control task as follows.
	\begin{problem}
		\label{prob:the_problem}
		Design a communication allocation which guarantees that the agent state $z$ remains inside an \textit{admissible} set $\m{A}\subset \m{X}\times\m{X}\times\m{U}$ at all time. 
	\end{problem}

	\subsection{Examples of models belonging to framework \eqref{eq:general_system}}
	The modeling structure introduced in the previous subsection is a general framework for modeling a number of different NCS with limited capacity in the communication links between controller, plant, and sensors.
	
	As a special case, the two dynamics in~\eqref{eq:general_system} could describe the evolution of a NCS, together with a predictor and a dynamic feedback controller. The connection signal~$\delta$, set by a central coordinator in a \emph{star} communication topology (e.g., cellular network), can describe (see Figure~\ref{fig:system_1})
		\begin{itemize}
			\item sensor-controller (SC) networks: $\delta_u=1, \ \delta_s =\delta$;
			\item controller-actuator (CA) networks: $\delta_u=\delta,\ \delta_s =1$;
			\item sensor-controller-actuator (SCA) networks: $\delta_u=\delta_s =\delta$.
	\end{itemize}
	For the three cases and in a linear setting, the functions~$F,~\hat{F}$ can be written as in the following examples.
	
	\begin{example}[Static state-feedback, SC network]
		\label{Exp:case i}
		\begin{equation}
		\label{eq:sc_network}
		F=\begin{pmatrix}
		Ax+BKx+Ev\\
		Ax+BKx\\
		\emptyset
		\end{pmatrix},~\hat{F}=\begin{pmatrix}
		Ax+BK\hat{x}+Ev\\
		A\hat{x}+BK\hat{x}\\
		\emptyset
		\end{pmatrix}.
		\end{equation}
	\end{example}
	
	\begin{example}[Static state-feedback, CA network]
		\begin{multline}
		F=\begin{pmatrix}
		\begin{pmatrix}
		A & 0\\
		0 & 0
		\end{pmatrix}x+
		\begin{pmatrix}
		B\\1
		\end{pmatrix}Kx+
		\begin{pmatrix}
		E\\0
		\end{pmatrix}v\\
		\emptyset\\
		\emptyset
		\end{pmatrix},\\
		\hat{F}=\begin{pmatrix}
		\begin{pmatrix}
		A & B\\
		0 & 1
		\end{pmatrix}x+
		\begin{pmatrix}
		E\\0
		\end{pmatrix}v\\
		\emptyset\\
		\emptyset
		\end{pmatrix},
		\end{multline}
		where $x$ is $n+1$ dimensional if $A$ is $n\times n$, and $K_{n+1}=0$. Element $n+1$ is a memory variable, used to implement a hold function on the last computed input, when no new state information is available.
	\end{example}
	
	\begin{example}[Dynamic state-feedback, SCA network]
		\begin{multline}
		F=\begin{pmatrix}
		\begin{pmatrix}
		A & 0\\
		0 & 0
		\end{pmatrix}x+
		\begin{pmatrix}
		B\\1
		\end{pmatrix}C_{c} u+
		\begin{pmatrix}
		E\\0
		\end{pmatrix}v\\
		Ax+BC_{c}u\\
		A_{c}u+B_{c}x
		\end{pmatrix},\\
		\hat{F}=\begin{pmatrix}
		\begin{pmatrix}
		A & B\\
		0 & 1
		\end{pmatrix}x+
		\begin{pmatrix}
		E\\0
		\end{pmatrix}v\\
		A\hat{x}+BC_{c}u\\
		A_{c}u+B_{c}\hat{x}
		\end{pmatrix},
		\end{multline}
		where $x$ is $n+1$ dimensional if $A$ is $n\times n$, and $K_{n+1}=0$. Element $n+1$ is a memory variable, used to implement a hold function on the last computed input, when no new state information is available.
	\end{example}
	
	\begin{figure}[t]
		\centering
		\includegraphics[width=0.5\columnwidth]{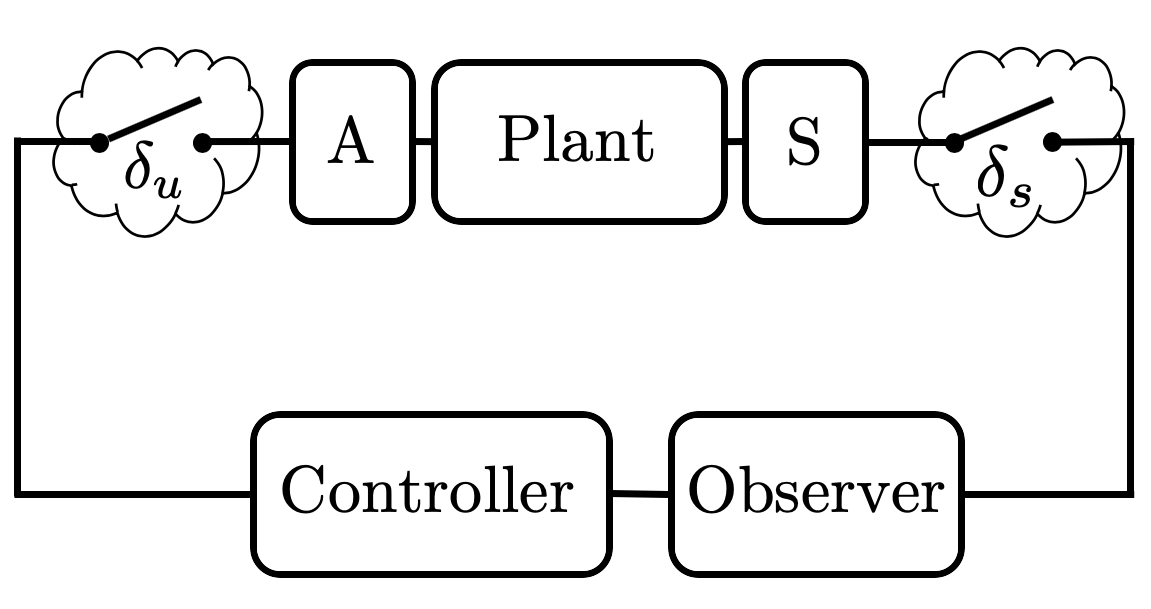}
		\caption{An example of a system described by the model~\eqref{eq:general_system}.}
		\label{fig:system_1}
	\end{figure}
	In all of the above cases, the decision variable $\delta$ is selected by a central scheduler. This might have access to the exact state, or might be subject to the same limitations on state information as the controller.  As we discuss in Section \ref{Sec_online_subsec_online}, the available information influences the scheduler design.
	
	\section{Mathematical background}
	\label{section:math-background}

	In order to translate control problem~\ref{prob:the_problem} into a scheduling problem, we will define the concept of \emph{safe time interval} by relying on robust invariance and reachability analysis.
	
	Set $\m{S}\subset\m{X}\times\m{X}\times\m{U}$ is \textit{robust invariant} for \eqref{eq:general_system} in connected mode, i.e., $\delta=1$ if
	\begin{align}
	F(z,v) \in \mathcal{S} ,&& \Forall z \in \m{S},~ v\in\mathcal{V}.
	\end{align}
	Any robust invariant set $\mathcal{S}$ contains all forward-time trajectories of the agent \eqref{eq:general_system} in the connected mode, provided $z(0) \in\m{S}$, regardless of the disturbance $v$. 
	
	Let $\{\m{S}\}$ be the set of all robust invariant sets of \eqref{eq:general_system} that are contained in the admissible set $\m{A}$. We call $\MR{} \in \{\mathcal{S}\}$ the \textit{maximal robust invariant} set: 
	\begin{align}
	\mathcal{S} \subseteq \mathcal{S}_{\infty},&& \forall \mathcal{S} \in \{\mathcal{S}\}.
	\end{align}
	
	We define the \emph{$1$-step reachable set}  as the set of  states $z$ that can be reached in \emph{one step} from a set of initial states~$\mathcal{O}$ with dynamics $H \in \{F,\hat{F}\}$:
	\begin{align}\label{eq:reach}
	\Reach{H}{1}{\mathcal{O}}:=\{H(z,v): ~ z \in \m{O}, ~ v \in \m{V}\}.
	\end{align}
	The $t$-step reachable set, $t=2,\ldots$ is defined recursively as
	\begin{align}\label{eq:reach_t_steps}
	\Reach{H}{t}{\mathcal{O}}:=\Reach{H}{1}{\Reach{H}{t-1}{\mathcal{O}}}.
	\end{align}
	Numerical tools for the calculation of~$\mathcal{S}_\infty$ and $\Reach{H}{t}{\mathcal{O}}$ can be found in~\cite{borrelli2017predictive}, for linear $H$.

	\begin{definition}[Safe time interval, from \cite{ECC}]
		\label{def:safe_time}
		We define the \emph{safe time interval} for agent $i$ as
		\begin{multline}\label{eq:safe_time}
		\alpha_i:=1+\max_{\tau} \Big{\{} \tau : \Reach{\hat{F}_i}{\tau}{\Reach{F_i}{1}{\MR{i}}} \subseteq \MR{i} \Big{\}}. 
		\end{multline}
	\end{definition}
	Essentially, $\alpha_i$ counts the amount of time steps during which agent $i$ can be disconnected while maintaining its state in $\MR{i}$, provided that its initial state is in $\MR{i}$.
	Note that, by definition of $\MR{i}$, agent $i$ remains in $\MR{i}$ for all future times when connected. 
	
	The following example illustrates the effect of measurement noise on the reachable set and on the safe time interval, in a system with static feedback structured as a SC network.
	\begin{example}\label{ex:alpha_sets}
		Consider an agent described by
		\begin{equation}
		x(t+1)=Ax(t)+Bu(t)+Ev(t)
		\end{equation}
		where
		\begin{equation}
		A=\begin{bmatrix}
		1 & 0.5\\-0.5 & 1
		\end{bmatrix}, \quad B=E=\begin{bmatrix}
		0\\1
		\end{bmatrix},
		\end{equation}
		with admissible sets
		\begin{equation}
		\mathcal{X}=\Big{\{} x \in \mathbb{R}^2: \begin{bmatrix}
		-2\\-2
		\end{bmatrix} \leq x \leq \begin{bmatrix}
		2\\2
		\end{bmatrix} \Big{\}}
		\end{equation}
		and
		\begin{equation}
		\mathcal{U}= \{u  \in \mathbb{R} : |u| \leq 5  \} ,~ \mathcal{V}=\{ v  \in \mathbb{R} : |v| \leq 0.45 \} ,
		\end{equation}
		and $u(t)=-K\hat{x}(t)$. The state is estimated according to
		\begin{equation}
		\hat{x}(t)=\begin{cases}
		x(t),& \text{if }C(t)=1\\
		A\hat{x}(t-1)+Bu(t-1), & \text{if } C(t)=0
		\end{cases}
		\end{equation}
		and the control gain is
		\begin{equation}
		K=\begin{bmatrix}
		0.2263 & 1.2988
		\end{bmatrix}.
		\end{equation}
		Consider a SA network, i.e., $\delta_u=1$ and $\delta_s(t)=C(t),~\forall t>0$. Furthermore, assume the agent is connected at $t=0$, i.e., $C(0)=1$, and disconnected afterwards, i.e., $C(t)=0, \forall t>0$. This implies $\hat{x}(0)=x(0)$ and
		\begin{equation}
		\hat{F}=\begin{bmatrix}
		Ax(t)-BK\hat{x}(t)+Ev(t)\\
		A\hat{x}(t)-BK\hat{x}(t)\\
		\emptyset
		\end{bmatrix}.
		\end{equation}
		The sets of states which can be reached in $t$ steps are displayed in Figure~\ref{fig:reachable_sets} where notation $\mathcal{X}_t$ indicates the reachable set for $x(t)$. 
		One can observe that in this example $\mathcal{X}_4 \not \subseteq \mathcal{S}_{\infty}$, which implies $\alpha=3$.
		\begin{figure}
			\centering
			\includegraphics[width=0.7\columnwidth]{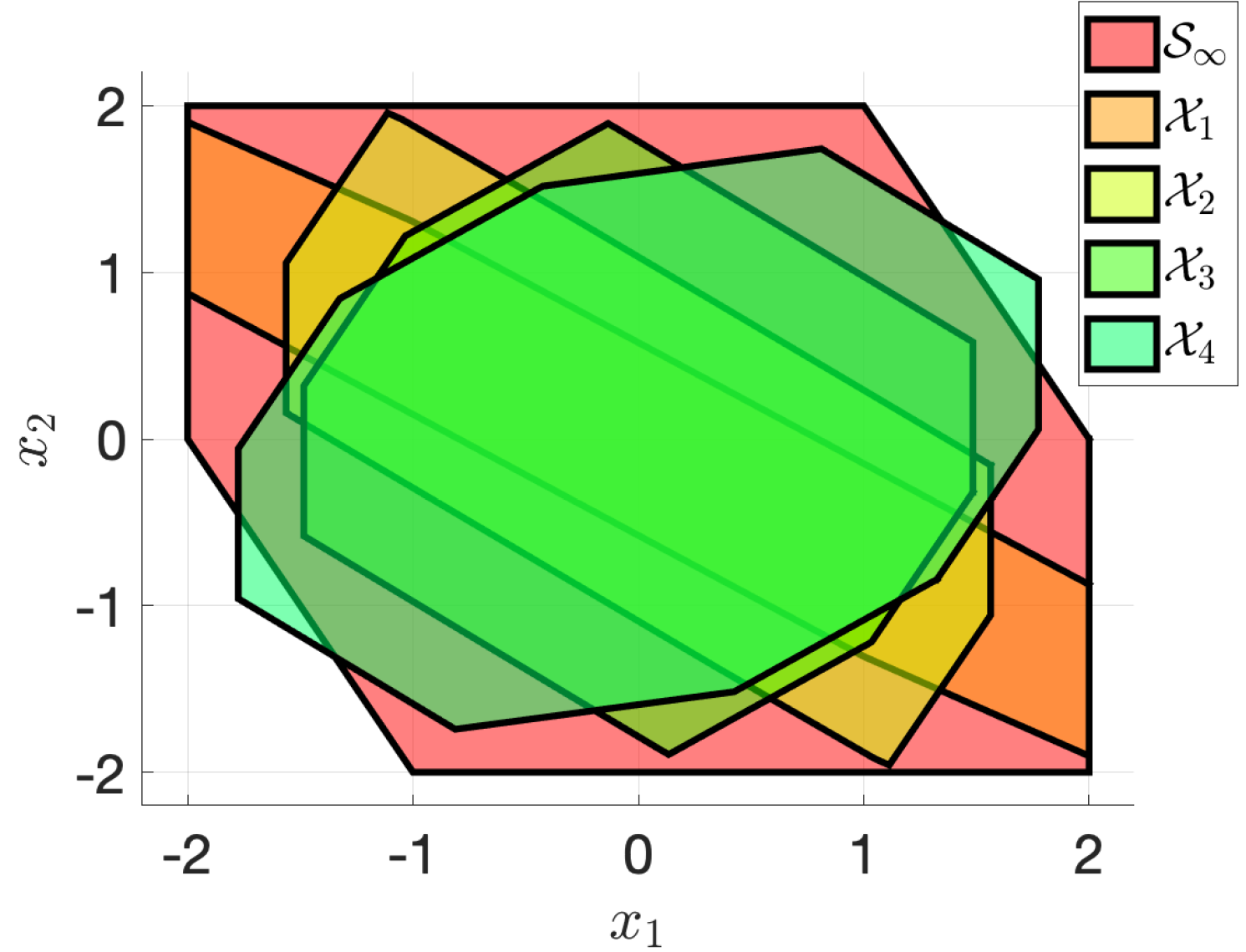}
			\caption
			{
				Maximum robust invariant set and $\mathcal{X}_t=\Reach{\hat{F}}{t}{\mathcal{S}_{\infty}}$ in Example~\ref{ex:alpha_sets} projected along $x$.
			}
			\label{fig:reachable_sets}
		\end{figure}
	\end{example}

	The task of keeping the state of each agent in its admissible set can now be formulated as follows.
	\begin{problem}[P\ref{pr:schedulability}]
		\label{pr:schedulability}
		Given the set of $q$ agents, each described by \eqref{eq:general_system}, an admissible set $\m{A}:= \m{A}_1\times\ldots\times\m{A}_q$, and the set $\boldsymbol{\mathcal{C}}$ of connection patterns \eqref{eq:binary_vectors}, determine if there exists an infinite sequence over the elements of $\boldsymbol{\mathcal{C}}$ such that, 
		\begin{multline}
		z_i(t) \in \MR{i},~\forall z_i(0) \in \MR{i}, \ v_i(t)\in\m{V}_i,\\
		i\in\{1,\ldots,q\},\ t>0.
		\end{multline}
	\end{problem}
	In other words, we seek an infinite sequence of connection patterns
	\begin{equation}
	\label{eq:feasibleSchedule}
	\mathbf{C}:= \big{(}C(0),~C(1),\ldots\big{)},
	\end{equation}
	with~$C(t)\in \boldsymbol{\mathcal{C}}$ that keeps $(x,\hat{x},u)$ of all $q$ agents within $\MR{}$ despite the fact that, due to the structure of set $\boldsymbol{\mathcal{C}}$, at each time step some agents are disconnected. Note that the set $\boldsymbol{\mathcal{C}}$ is assumed to be fixed and given \textit{a priori}, e.g., based on the network structure.

	A \textit{schedule} solving P\ref{pr:schedulability} is any sequence of $C_j$ such that every agent $i$ is connected at least once every $\alpha_i$ steps. Instance $I:=\{\boldsymbol{\mathcal{C}},\{\alpha_i\} \}$ is accepted, denoted
	\begin{equation}
	\label{eq:instanceP1}
	I \in \text{P\ref{pr:schedulability}},
	\end{equation}
	if and only if a schedule $\mathbf{C}$ exists that satisfies P\ref{pr:schedulability}.  
	
	In order to find a feasible schedule for P\ref{pr:schedulability}, we will translate this problem into a PP or WSP. In this section, we formally introduce these two problems and discuss their fundamental properties.
	
	\subsection{The Pinwheel Problem}
	\begin{PinwheelProblem}[From \cite{Han96a}]
		Given a set of integers $\{\alpha_i\}$ with $\alpha_i\geq 1$, determine the existence of an infinite sequence of the symbols $1,\ldots,q$ such that there is at least one symbol $i$ within any subsequence of $\alpha_i$ consecutive symbols. 
	\end{PinwheelProblem}
	A schedule solving PP can be defined by using the notation in Section~\ref{section:problem-statement}, as
	\begin{equation}
	\mathbf{C}:=c(1),c(2),\ldots~,
	\end{equation}
	with $c(t)\in \{1,\ldots,q\}$.
	\begin{definition}[Feasible schedule]
		\label{def:Schedule}
		A schedule $\mathbf{C}$ that solves a schedulability problem is called a \textit{feasible schedule} for that problem.
	\end{definition}
	
	Instance $I:=\{\alpha_i\}$ is accepted by PP, denoted
	\begin{equation}
	\label{eq:instancePP}
	I\in \text{PP},
	\end{equation}
	if and only if a feasible schedule $\mathbf{C}$ exists for the problem.
	
	Conditions for schedulability, i.e., existence of a feasible solution of PP, have been formulated in terms of the $\textit{density}$ of a problem instance $I$, defined as
	\begin{equation}
	\rho(I):=\sum_i \frac{1}{\alpha_i}.
	\end{equation}
	\begin{theorem}[Schedulability conditions]
		\label{thm:schedulability_threshold}
		Given an instance $I:=\{\alpha_i\}$ of PP, 
		\begin{enumerate}
			\item if $\rho(I)> 1$ then $I \notin \text{PP}$,
			\item if $\rho(I)\leq 0.75$ then $I \in \text{PP}$,
			\item if $\rho(I)\leq \frac{5}{6}$ and there exists $i:\alpha_i=2$ then $I \in \text{PP}$,
			\item if $\rho(I)\leq \frac{5}{6}$ and $I$ has only three symbols then $I \in \text{PP}$,
			\item if $\rho(I)\leq 1$ and $I$ has only two symbols $\alpha_1$ and $\alpha_2$ then $I \in \text{PP}$.
		\end{enumerate}
	\end{theorem}
	\begin{proof}
		Condition $1$ is a simple consequence of the definition of density: if $\rho(I)>1$ there are not enough time slots to schedule all symbols $\{\alpha_i \}$.  Conditions $2$ and $3$ are proved in \cite{Fishburn02a}.  Conditions $4$ and $5$ are proved in \cite{Chen04a}.
	\end{proof}
	It has been conjectured that any instance of PP with $\rho(I) \leq \frac{5}{6}$ is schedulable; however, the correctness of this conjecture has not been proved yet \cite{Chan92a}. Determining whether a general instance of PP with $\frac{5}{6}<\rho(I) \leq 1$ is schedulable, is not possible just based on the density $\rho(I)$ (e.g., $\rho(\{2,2\})=1$ is schedulable and $\rho(\{2,3,12\})=\frac{11}{12}$ is not schedulable). Furthermore, determining the schedulability of dense instances, i.e., when $ \rho(I) = 1$, is NP-hard in general \cite{Holte89a}.
	
	Since a schedule for PP is an infinite sequence of symbols, the scheduling search space is also infinite dimensional. Fortunately, the following theorem alleviates this issue. 
	\begin{theorem}[Theorem 2.1 in \cite{Holte89a}]
		All instances of PP that admit a schedule admit a \textit{cyclic schedule}, i.e., a schedule whose symbols repeat periodically.
	\end{theorem}  
	
	\subsection{The Windows Scheduling Problem}
	WSP is a more general version of PP, where multiple symbols can be scheduled at the same time. We call \textit{channels} the multiple strings of symbols that constitute a Windows Schedule. 
	\begin{WindowsSchedulingProblem}[From \cite{bar2003windows}]
		Given the set of integers $\{\alpha_i\}$ with $\alpha_i\geq 1$, determine the existence of an infinite sequence of ordered tuples with $m_{\mathrm{c}} \geq 1$ elements of the set $\{1,\ldots,q\}$ such that there is at least one tuple that contains the symbol $i$ within any subsequence of $\alpha_i$ consecutive tuples. 
		
		An instance $\{m_{\mathrm{c}},\{\alpha_i\}\}$ of WSP  is accepted, and denoted as
		\begin{equation}
		\label{eq:instanceWSP}
		I=\{m_{\mathrm{c}},\{\alpha_i\}\} \in \text{WSP},
		\end{equation}
		if and only if a feasible schedule 
		\begin{equation}
		\mathbf{C}=C(1),C(2),\ldots
		\end{equation}
		with
		\begin{equation}
		C(t)=\left( c_1(t),\ldots,c_{m_{\mathrm{c}}}(t) \right)
		\end{equation}
		exists for the problem.
		
	\end{WindowsSchedulingProblem}
	WSP is equivalent to PP when $m_{\mathrm{c}}=1$.  Similarly to PP, if a schedule for the WSP exists, then a cyclic schedule exists.
	Furthermore, the following schedulability conditions are known.
	\begin{theorem}
		\label{thm:schedulability_thresholdmc}
		Given an instance $I=\{m_{\mathrm{c}},\{\alpha_i\}\}$ of WSP, 
		\begin{enumerate}
			\item if $\rho(I)> m_{\mathrm{c}}$ then $I \notin \text{WSP}$ ,
			\item if $\rho(I)\leq 0.5m_{\mathrm{c}}$ then $I \in \text{WSP}$.
		\end{enumerate}
	\end{theorem}
	\begin{proof}
		Condition 1) is a direct consequence of the definition of schedule density; condition 2) is proved in Lemma 4 and 5 in \cite{bar2003windows}.
	\end{proof}
	
	The results on WSP used next rely on special schedules of a particular form, defined as follows.
	
	\begin{definition}[Migration and perfect schedule, from \cite{bar2003windows,bar2007windows}]
		A \textit{migrating} symbol is a symbol that is assigned to different channels at different time instants of a schedule.  A schedule with no migrating symbols is called a \textit{perfect schedule}.   
	\end{definition}
	An instance $I:=\{m_{\mathrm{c}},\{\alpha_i\}\}$ of WSP is accepted with a perfect schedule if and only if a feasible schedule $\mathbf{C}$ exists for the problem such that
	\begin{equation}
	c_i(t_1)=c_k(t_2) \implies i=k,
	\label{eq:perfect_schedule_condition}
	\end{equation}
	for any $ i,k \in \{1,\ldots,m_{\mathrm{c}}\}$ and $t_1,t_2 \in \mathbb{N}$; we denote this as:
	\begin{equation}
	\label{eq:instanceWSP_Pf}
	I \in \text{WSP-perfect}.
	\end{equation}
	Equation~\eqref{eq:perfect_schedule_condition} ensures that agents do not appear on different channels of the schedule.
	\section{Main Results: offline scheduling}
	\label{section:results_offline}
	In this section, we provide theoretical results and algorithms to solve P\ref{pr:schedulability}. In Subsection~\ref{subsec_V_A}, P\ref{pr:schedulability} is considered in the most general form and we prove that the problem is decidable, i.e., there is an algorithm that determines whether an instance is accepted by the problem \cite{margenstern2000frontier}. In Subsection~\ref{subsec_V_B}, we provide a heuristic to find a feasible schedule. In the last subsection, we consider a fixed number of communication channels. In this case, we show that the scheduling problem is equivalent to the WSP. We propose a technique to solve the scheduling problem in this case and illustrate the merits of the proposed heuristic with respect to the existing ones. We also refute a standing conjecture regarding perfect schedules in WSP~\cite{bar2007windows}.
	
	\subsection{Solution of P\ref{pr:schedulability}}
	\label{subsec_V_A}
	In this subsection we show that P\ref{pr:schedulability} is decidable by showing that if there exists a feasible schedule for the problem, then there also exists a cyclic schedule with bounded period. Finally, we provide an optimization problem to find a feasible cyclic schedule.
	
	Consider sequence $\mathbf{C}$ as the schedule for P\ref{pr:schedulability}, and define sequence 
	\begin{equation}
	\mathbf{D}:=D(1),D(2),\ldots
	\label{eq:def_D}
	\end{equation}
	with the vector $D(t)$ defined as
	\begin{equation}
	\label{Ctilde}
	D(t):=\left( d_1(t),d_2(t),\ldots,d_q(t)\right),
	\end{equation}
	where $d_i(t):=t-\tau_i^{\mathbf{C}}(t)$, and the \textit{latest connection time} $\tau_i^{\mathbf{C}}(t)$ is defined as:
	\begin{equation}
	\tau_i^{\mathbf{C}}(t):=\max \{t^\prime\le t:~ i \in C(t^\prime)\},
	\label{eq:TauLastMeasurement}
	\end{equation}
	where $t^\prime:=0$ when the above set is empty.
	\begin{lemma}
		\label{schedule_Cbar}
		The schedule $\mathbf{C}$ is feasible for P\ref{pr:schedulability} if and only if $0 \leq d_i(t) \leq \alpha_i -1$, $\forall i \in \{1,\ldots,q\},~ \forall t>0$.
	\end{lemma}
	\begin{proof}
		If $\mathbf{C}$ is a feasible schedule, then $0 \leq d_i(t) \leq \alpha_i -1$ for $\forall i \in \{1,\ldots,q\},~ \forall t>0$ by construction. This implies that agent $i$ is connected at least once every $\left( 1+\max_t d_i(t) \right) \leq \alpha_i$ time instants. Therefore, $\mathbf{C}$ is a feasible schedule.
	\end{proof}
	
	\begin{theorem}
		\label{thm:Cyclic_schedulability}
		Consider the set of integers $\{ \alpha_i \}$  defined in \eqref{eq:safe_time}.  If P\ref{pr:schedulability} has a feasible schedule $\mathbf{C}$, then it has a cyclic schedule whose period is no greater than $m=\prod_{i=1}^q \alpha_i$.
	\end{theorem}
	\begin{proof}
		We define $\mathbf{D}$ as in \eqref{Ctilde}, so that  $ 0 \leq d_i(t) \leq \alpha_i -1 $ holds by Lemma~\ref{schedule_Cbar}. Hence, each $d_i(t)$ can have no more than $\alpha_i$ different values. This implies $D(t)$ can have at most $m~:=~\prod_{i=1}^q \alpha_i$ different values. Hence,
		\begin{equation}
		\exists ~t_1,t_2:~D(t_1)=D(t_2),~m \leq t_1 < t_2 < 2m.
		\end{equation}
		Now, consider the sequence 
		\begin{equation}
		\mathbf{C}_{\mathrm{r}}:=\big{(}C(t_1),C(t_1+1),\ldots,C(t_2-1)\big{)}
		\end{equation}
		as the cyclic part of the cyclic schedule  $\mathbf{C}_{\mathrm{c}}$ for P\ref{pr:schedulability}, defined as
		\begin{equation}
		\mathbf{C}_c:=\big{(} \mathbf{C}_{\mathrm{r}}, \mathbf{C}_{\mathrm{r}},\ldots\big{)}.
		\end{equation}
		Define $\mathbf{D}_{\mathrm{c}}$ as in \eqref{eq:def_D} for the new schedule $\mathbf{C}_{\mathrm{c}}$. One can conclude that
		\begin{equation}
		D_{\mathrm{c}}(\tau) \leq  D(\tau+t_1-1), \qquad \forall \tau \in \{1,\ldots,t_2-t_1\},
		\end{equation}
		since for any $ i \in \{1,\ldots,q\}$ we have
		\begin{equation}
		{d_{\mathrm{c}}}_i(\tau)=\tau-\tau_i^{\mathbf{C}_{\mathrm{c}}}=(\tau+t_1-1)-(\tau_i^{\mathbf{C}_{\mathrm{c}}}+t_1-1)\leq d_i(\tau+t_1-1).
		\end{equation}
		Furthermore, $D(t_1)=D(t_2)$ implies $i \in \mathbf{C}_{\mathrm{r}}$ for $\forall i \in \{1,\ldots,q\}$. As a result, ${d_{\mathrm{c}}}_i(t_2-t_1)=d_i(t_2-1)$. This implies
		\begin{equation}
		{d_{\mathrm{c}}}_i(k(t_2-t_1)+\tau)=d_i(t_1-1+\tau), \qquad k \in \mathbb{N}.
		\end{equation}
		Since $d_i(t)\leq \alpha_i-1$ holds for any $t>0$, then ${d_{\mathrm{c}}}_i(t) \leq \alpha_i-1$ also holds for any $t>0$. Inequality ${d_{\mathrm{c}}}_i(t) \leq \alpha_i-1$ implies that $\mathbf{C}_c$ is a feasible schedule by Lemma~\ref{schedule_Cbar}.
	\end{proof}
	
	Theorem~\ref{thm:Cyclic_schedulability} implies that a feasible schedule can always be searched for within the finite set of cyclic schedules of a length no greater than $m$. An important consequence of this theorem is the following.
	\begin{theorem}
		P\ref{pr:schedulability} is decidable.
	\end{theorem}
	\begin{proof}
		Since the search space is a finite set, schedules can be finitely enumerated.  
	\end{proof}
	Theorem~\ref{thm:Cyclic_schedulability} allows us to solve P\ref{pr:schedulability} by solving the following optimization problem, which searches for a feasible periodic schedule among all schedules of period $T_r$.
	\begin{subequations}
		\label{eq:qnary_formulation_exact_solution}
		\begin{align}
		\min_{C(1),\ldots,C(T_{\mathrm{r}}),T_{\mathrm{r}}} & \ \ T_{\mathrm{r}}\\
		\mathrm{s.t.}& \quad C(1),\ldots,C(T_{\mathrm{r}}) \in \boldsymbol{\mathcal{C}}, \label{eq:qnary_formulation_exact_solution_b}\\
		&\quad  T_{\mathrm{r}} \leq \prod_{i=1}^q \alpha_i,~T_{\mathrm{r}} \in \mathbb{N}, \label{eq:qnary_formulation_exact_solution_c}\\		
		& \quad \sum_{k=t}^{t+\alpha_j-1} \eta_j ( k ) \geq 1, \label{eq:qnary_formulation_exact_solution_d}\\
		& \quad \forall j \in \{1,\ldots,q\},~ \forall t \in \{1,\ldots,T_{\mathrm{r}}\}, \nonumber\\
		&\quad \eta_j(k) =\begin{cases}
		1 &  \text{if } j\in C(k \bmod T_{\mathrm{r}}),\\
		0 & \textrm{ otherwise}.
		\end{cases} \label{eq:qnary_formulation_exact_solution_e}
		\end{align}
	\end{subequations}
	Note that we define $C(0):=C(T_{\mathrm{r}})$ in \eqref{eq:qnary_formulation_exact_solution_e}. Equation \eqref{eq:qnary_formulation_exact_solution_b} enforces the schedule elements to be chosen from the set of connection patterns $\boldsymbol{\mathcal{C}}$; \eqref{eq:qnary_formulation_exact_solution_c} limits the search space by giving an upper bound for the length of the periodic part, i.e., $T_{\mathrm{r}}$; and \eqref{eq:qnary_formulation_exact_solution_d} ensures that label $i$ appears at least once in each  $\alpha_i$ successive elements of the schedule sequence.
	
	Note that the main challenge in Problem~\eqref{eq:qnary_formulation_exact_solution} is finding a feasible solution; minimization of $T_{\mathrm{r}}$ is a secondary goal since any solution of \eqref{eq:qnary_formulation_exact_solution} provides a feasible schedule for P\ref{pr:schedulability}. Unfortunately, \eqref{eq:qnary_formulation_exact_solution_e} is combinatorial in the number of agents and connection patterns.
	In order to tackle this issue, we next propose a strategy to simplify the computation of a feasible schedule.
	
	\subsection{A heuristic solution to P\ref{pr:schedulability}}
	\label{subsec_V_B}
	In this subsection, we propose a heuristic to solve P\ref{pr:schedulability} based on the assumption that the satisfaction of the constraints for a agent $i$ is a duty assigned to a single connection pattern $C_j$.
	
	To give some intuition on the assignment of connection patterns, we propose the following example.
	\begin{example} 
		\label{ex:toy_general_case}
		\begin{figure}[h]
			\centering
			\includegraphics[width=0.4\columnwidth]{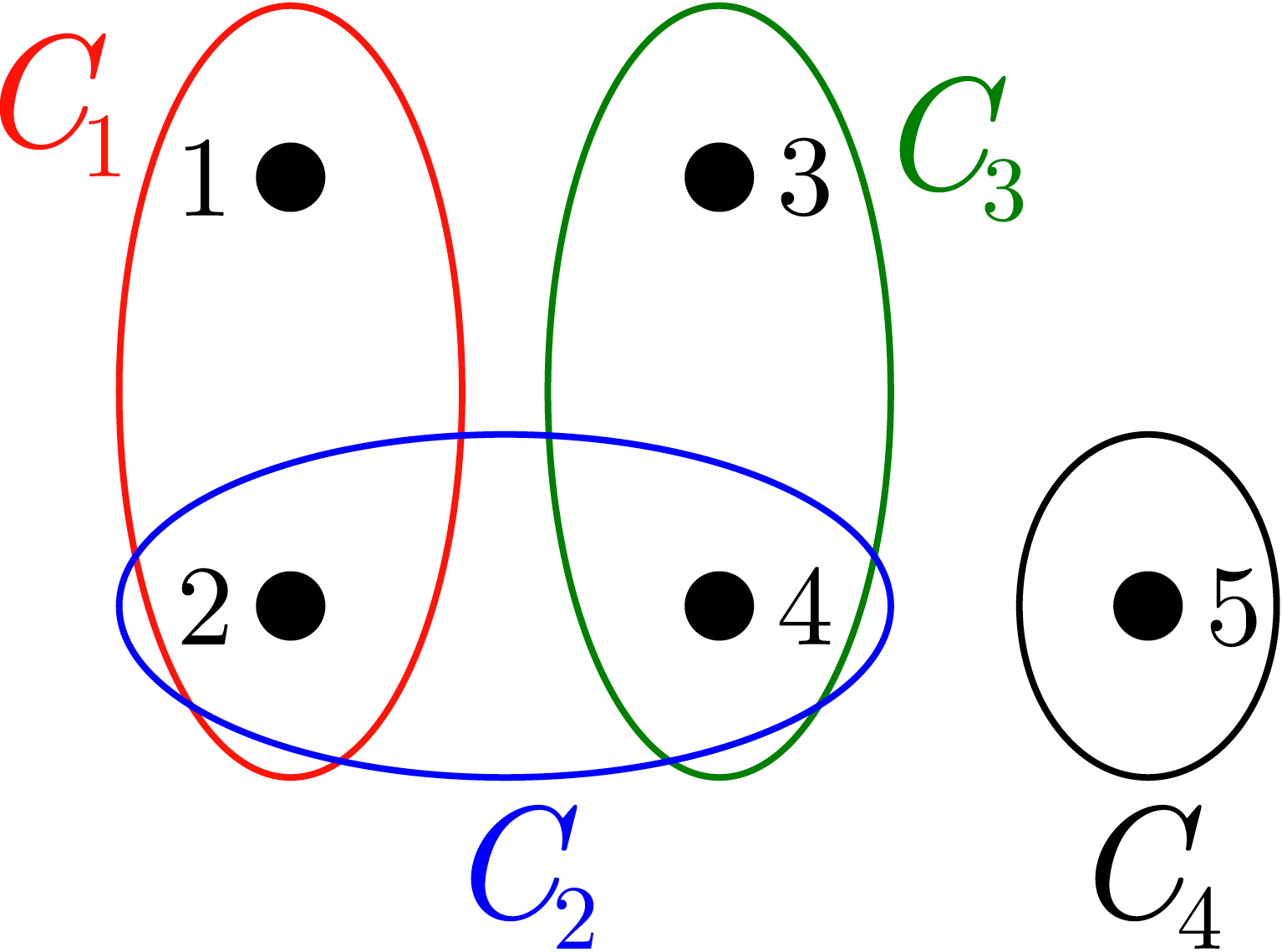}
			\caption{A model with $5$ agents and $4$ connection patterns.}
			\label{fig:connection_pattern_example}
		\end{figure}
		Consider the network displayed in Figure~\ref{fig:connection_pattern_example}, with five agents which can be connected according to $4$ connection patterns: $C_1:=(1,2)$, $C_2:=(2,4)$, $C_3:=(3,4)$, $C_4:=(5)$.
		Assume that the safe time intervals of the $5$ agents are $\alpha_1=10$, $\alpha_2=2$, $\alpha_3=10$, $\alpha_4=2$, $\alpha_5=100$. This means that agents $2$ and $4$ must be connected at least once every $2$ steps, while the other agents have less demanding requirements.   
		
		As a first try, let us attempt a schedule using only connection patterns $C_1$, $C_3$, $C_4$. 
		In this case, one can see that the sequence $(C_1,C_3,C_1,C_3,\ldots)$ is the only possible schedule satisfying the requirements of agents $2$ and $4$. There is however no space to connect agent $5$ within this schedule. As an alternative solution we therefore propose to utilize the patterns $C_1$, $C_2$, $C_3$, $C_4$ and design $(C_2,C_1,C_2,C_3,C_2,C_4)$ as the cyclic part of a schedule. One can verify that this schedule is feasible.
		
		With the first choice of connection patterns, the duty of satisfying the constraints for agents $2$ and $4$ is assigned to the patterns~$C_1$ and $C_3$, respectively, which, therefore, must be scheduled every 2 steps. On the other hand, with the second choice, this duty is assigned to~$C_2$, while agents~$1$ and $3$ are assigned to~$C_1$ and $C_3$, respectively. As a consequence, $C_2$ must be scheduled every $2$ steps, but~$C_1$ and $C_3$ can be scheduled once every $10$ steps. This allows one to make space for $C_4$. Borrowing the terminology of PP, with the first choice~$C_1$ and $C_3$ are symbols of density $0.5$ and $C_4$ has density $0.01$. Hence, the three symbols are not schedulable. With the second choice, instead, $C_2$ has density $0.5$, $C_1$ and $C_3$ have density $0.1$, and $C_4$ has density $0.01$. Hence, the total density is $0.71$ and the four symbols are schedulable. 
	\end{example}
	
	Example \ref{ex:toy_general_case} shows how we assign duties to the connection patterns, and also how 
	schedulability is affected. In the following, we formulate a problem that selects the connection patterns in order to minimize the total density.
	
	Let us now represent the assignment of agent $i$ to the connection pattern $C_j$ with a binary variable $\eta_{i,j}$ and---with a slight abuse of notation---the density of symbol $C_j$ with $\hat{\rho}_j$. The proposed strategy is to decide the set of $\eta_{i,j}$ such that $\sum_j \hat{\rho}_j$ is minimized. This is performed by solving
	\begin{subequations}%
		\label{eq:min_connection}
		\begin{align}
		& \min_{\hat{\rho}_j,~\eta_{i,j}} \sum_{j=1}^l \hat{\rho}_j  \label{eq:min_connection_cost}\\
		\mathrm{s.t.} \quad &\hat{\rho}_j \geq \frac{1}{\alpha_i} \eta_{i,j}, \ \forall j \in \{1,\ldots,l\}, \ \forall i \in C_j, \label{eq:min_connection_beta}\\
		& \sum_{j:i \in C_j}\eta_{i,j} \geq 1, \ \forall i \in \{1,\ldots,q\}, \label{eq:min_connection_eta}\\
		& \eta_{i,j} \in \{0,1\}, \ \forall i \in \{1,\ldots,q\}, \ \forall j \in \{1,\ldots,l\}.
		\end{align}
	\end{subequations}
	Constraint~\eqref{eq:min_connection_eta} guarantees that every agent~$i$ is connected by at least one connection pattern. Variables~$\hat{\rho}_j$ bound the density of the resulting scheduling problem, where~$1/\hat{\rho}_j$ is the maximum number of steps between two occurrences of connection pattern~$C_j$ in~$\mathbf{C}$ that is sufficient to enforce $(x_i,\hat{x}_i,u_i) \in \MR{i}$. If $\hat{\rho}_j=0$, then connection pattern $j$ is not used.  Without loss of generality, assume that the solution to \eqref{eq:min_connection} returns $l$ distinct connection patterns with $\hat{\rho}>0$, i.e., $\hat{\rho}_1,\ldots,\hat{\rho}_l>0$ and define
	\begin{equation}
	\hat{\alpha}_i:=\frac{1}{\hat{\rho}_i},  \qquad \forall i \in \{1,\ldots,l\}.
	\label{eq:new_alpha_GWSP}
	\end{equation}
	
	\begin{theorem}
		\label{thm:schedulability_thresholdmcS}
		$\{\hat{\alpha}_i\} \in \text{PP} \implies \{\boldsymbol{\mathcal{C}},\{\alpha_i\} \} \in \text{P\ref{pr:schedulability}}$.
	\end{theorem}
	\begin{proof}
		Consider any schedule $\mathbf{C}_{\mathrm{P}}$ which is feasible for the instance $\{ \hat{\alpha}_i \}$ of PP. 
		Define the schedule $\mathbf{C}$ by $C(t):=C_j$ given ${c_{\mathrm{P}}}(t)=j$ for any $j$. By the statement of PP, ${c_{\mathrm{P}}}(t)=j$ once at least in every \(\hat{\alpha}_j\) successive time instants. By \eqref{eq:min_connection}, for all $i$ there exists \(C_j\) such that $i \in C_j$ and \(\alpha_i \geq \hat{\alpha}_i\). Hence, $\mathbf{C}$ is a feasible schedule for P\ref{pr:schedulability}.
	\end{proof}
	
	Using Theorem~\ref{thm:schedulability_thresholdmcS}, we propose the following algorithm to find a feasible schedule for P\ref{pr:schedulability}.
	
	\begin{algorithm}
		\begin{algorithmic}[1]
			\footnotesize
			\State Define $\hat{\alpha}_i$ as in \eqref{eq:new_alpha_GWSP} by solving the optimization problem~\eqref{eq:min_connection}
			\If {$\{\hat{\alpha}_i\} \in \text{PP}$}
			\State find a schedule $\mathbf{C}_{\mathrm{P}}$ for instance $\{\hat{\alpha}_i\}$ of PP using \eqref{eq:qnary_formulation_exact_solution} or any other suitable scheduling technique
			\State define $C(t) := C_j$ given $C_\mathrm{P}(t)= j$
			\State \Return {$\mathbf{C}:= \big{(}C(1),~C(2),\ldots\big{)}$}
			\Else
			\State \Return{no schedule was found}
			\EndIf
		\end{algorithmic}
		\caption{A heuristic scheduling for P\ref{pr:schedulability} (\textbf{input}: $\{\{\alpha_i\}, \boldsymbol{\mathcal{C}}\}$, \textbf{output}: $\boldsymbol{C}$)}
		\label{alg:General_P1}
	\end{algorithm}
	
	As shown by the following example, the converse of Theorem~\ref{thm:schedulability_thresholdmcS} does not hold in general. I.e., if Algorithm~\ref{alg:General_P1} does not find a schedule a feasible schedule may still exist for P\ref{pr:schedulability}.
	\begin{example}[Converse of Theorem~\ref{thm:schedulability_thresholdmcS}] \label{ex:conv1}
		Consider five agents with \(\alpha_1=\alpha_3=3,\ \alpha_2=\alpha_4=\alpha_5=5\) and $\boldsymbol{\mathcal{C}}:=\{C_1,C_2,C_3,C_4\}$ where
		\begin{equation}
		C_1=(1,2),~C_2=3,~C_3=4,~C_4=(1,5).
		\end{equation}
		Using \eqref{eq:min_connection}, one obtains
		\begin{equation}
		\hat{\alpha}_1=\hat{\alpha}_3=5, \quad \hat{\alpha}_2=\hat{\alpha}_4=3.
		\end{equation}
		There is no feasible schedule for this problem considering the assigned density function $\hat{\rho}(\{3,3,5,5\})=\frac{16}{15}$, see Theorem~\ref{thm:schedulability_threshold}. However, one can verify that
		the following schedule is feasible
		\begin{equation}
		\mathbf{C}_c:=\big{(}\mathbf{C}_{\mathrm{r}},\mathbf{C}_{\mathrm{r}},\ldots \big{)},
		\end{equation}
		where
		\begin{equation}
		\mathbf{C}_{\mathrm{r}}:=\big{(}C_1, C_2, C_4, C_2, C_3 \big{)}.
		\end{equation}
	\end{example}
	\subsection{Solution of P\ref{pr:schedulability} in the $m_{\mathrm{c}}$-channel case}
	\label{subsection:WSP}
	In the previous subsection, $\boldsymbol{\mathcal{C}}$ was an arbitrary set of connection patterns. Assume now that the set~$\boldsymbol{\mathcal{C}}$ is 
	\begin{equation}
	\label{eq:mc_channels}
	\boldsymbol{\mathcal{C}}:=\{C:~C \subseteq \{1,\ldots,q\},~ |C|=m_{\mathrm{c}}\},
	\end{equation} 
	i.e., the set of all subsets of $ \{1,\ldots,q\}$ with cardinality $m_{\mathrm{c}}$. This is a special case of P\ref{pr:schedulability} where any combination of $m_{\mathrm{c}}$ agents can be connected at the same time. One application of such case is for instance when the connection patterns model a multi-channel star communication topology between a set of agents and a central controller. This class of problems is easily mapped to the class of WSP:
	\begin{theorem}
		\label{thm:equivalencemc}
		When $\boldsymbol{\mathcal{C}}$ is as in \eqref{eq:mc_channels}, then 
		\begin{equation}
		\{\boldsymbol{\mathcal{C}},\{\alpha_i\} \} \in \text{P\ref{pr:schedulability}} \iff \{m_c,\{\alpha_i\} \} \in \text{WSP}.
		\end{equation}
	\end{theorem}
	\begin{proof}
		By definition, any schedule solving P\ref{pr:schedulability} must satisfy $i \in C(t)\Rightarrow i \in C(t+\tau)$ with $\tau\leq \alpha_i$ for all $i,t$. Provided that $|C(t)|=m_\mathrm{c}$ for all $t$, this also defines a schedule solving WSP.
	\end{proof}
	We exploit this result to solve P\ref{pr:schedulability} indirectly by solving WSP. To that end, we propose a heuristic which replaces WSP with a PP relying on modified safe time intervals defined as
	\begin{equation}
	\label{new_alpha}
	\tilde{\alpha}_i:=m_{\mathrm{c}} \alpha_i,  \qquad  \forall i \in  \{1,\ldots,q\}.
	\end{equation}
	\begin{theorem}
		\label{malpha}
		$\{\tilde{\alpha}_i\} \in \text{PP} \implies \{m_{\mathrm{c}},\{\alpha_i\}\} \in \text{WSP}$.
	\end{theorem}
	\begin{proof}
		Given a feasible schedule $\mathbf{C}_{\mathrm{P}}$ for PP, ${{c_{\mathrm{P}}}(t)=i}$ at least once every \(\tilde{\alpha}_i=m_{\mathrm{c}}\alpha_i\) successive time instants. Define schedule $\mathbf{C}$ as
		\begin{equation}
		C(t):=\left( c_{\mathrm{P}}(m_{\mathrm{c}}(t-1)+1),\ldots, c_{\mathrm{P}}(m_{\mathrm{c}}t) \right).
		\end{equation}
		In this schedule, $i \in C(t)$ at least once every $\alpha_i$ successive time instants. This implies that $\mathbf{C}$ is a feasible schedule for WSP.
	\end{proof}
	
	Theorem~\ref{malpha} can be used to find a feasible schedule for WSP using a feasible schedule for PP. The converse of this theorem does not hold: if this method does not find a feasible schedule, 
	a feasible schedule for WSP may still exist. Nevertheless, Lemma~\ref{mc_lemma_1}  provides a sufficient condition to determine when a feasible schedule for WSP does not exist. Without loss of generality, assume \(\alpha_1 \leq \alpha_2 \leq \ldots \leq \alpha_q \) and define
	\begin{equation}
	\zeta_i:=\begin{cases}
	m_{\mathrm{c}} \alpha_i & i \leq m_{\mathrm{c}}\\
	m_{\mathrm{c}} \alpha_i+(m_{\mathrm{c}}-1) & i > m_{\mathrm{c}}
	\end{cases}.
	\label{converse_alpha}
	\end{equation}
	\begin{lemma}
		\label{mc_lemma_1}
		$ \{\zeta_i\} \notin \text{PP} \implies \{m_{\mathrm{c}},\{\alpha_i\}\} \notin \text{WSP}.$
	\end{lemma}
	\begin{proof}
		We proceed by contradiction. Assume $\{m_{\mathrm{c}},\{\alpha_i\}\} \in \text{WSP}$ with a corresponding feasible schedule $\mathbf{C}$, while $\{\zeta_i\}~\notin~\text{PP}$.
		Without loss of generality, assume that the labels $i \in \{1,\ldots,m_{\mathrm{c}}\}$ are arranged in $C(t)$ so as to satisfy
		\begin{equation}
		\label{eq:orderC}
		i \in C(t) \implies c_i(t)=i,
		\end{equation}
		while labels $i \in \{m_{\mathrm{c}},\ldots,q \}$ are arranged in an arbitrary order.
		Using the ordered $C(t)$, construct a schedule $\mathbf{C}_{\mathrm{P}}$ as
		\begin{equation}
		\label{eq:CP1}
		\mathbf{C}_{\mathrm{P}}=\left(c_1(1),\ldots,c_{m_{\mathrm{c}}}(1),\ldots,c_1(t),\ldots,c_{m_{\mathrm{c}}}(t),\ldots \right).
		\end{equation}
		If $\{\zeta_i\} \notin \text{PP}$, then there exists a $t_0>0$ and an $i \in  \{1,\ldots,q\}$ such that the sequence $\left( {c_{\mathrm{P}}}(t_0),\ldots,{c_{\mathrm{P}}}(t_0+\zeta_i-1)\right)$ does not contain label $i$, where $c_{\mathrm{P}}(t_0)$ is the entry $c_k(t)$ of $\mathbf{C}$. A pair of integers $(t,k)$ can be found such that $t \geq 0$, $k \in \{1,\ldots,m_{\mathrm{c}}\}$, and 
		\begin{equation}
		t_0=m_{\mathrm{c}}(t-1)+k.
		\end{equation}
		Consider the case $i \in \{1,\ldots,m_{\mathrm{c}}\}$. By \eqref{converse_alpha} we have $\zeta_i=m_{\mathrm{c}}\alpha_i$, such that the sequence $\left( {c_{\mathrm{P}}}(t_0),\ldots,{c_{\mathrm{P}}}(t_0+\zeta_i-1)\right)$ contains exactly $\alpha_i$ vectors $C(t),\ldots,C(t+\alpha_i-1)$ if $k=1$, or spans $\alpha_i+1$ vectors $C(t),\ldots,C(t+\alpha_i)$ if $k>1$. Hence, if $k\leq i$, by \eqref{eq:orderC} one can conclude $i \notin \left(C(t),\ldots,C(t+\alpha_i-1) \right)$, while if $k>i$, by \eqref{eq:orderC}, $i \notin \left(C(t+1),\ldots,C(t+\alpha_i) \right)$. In both cases $\mathbf{C}$ is not feasible, which contradicts our assumption. 
		\\
		Consider the case $i \in \{m_{\mathrm{c}},\ldots,q\}$. By \eqref{converse_alpha} we have $\zeta_i=m_{\mathrm{c}}(\alpha_i+1)-1$, such that the sequence $\left( {c_{\mathrm{P}}}(t_0),\ldots,{c_{\mathrm{P}}}(t_0+\zeta_i-1)\right)$ contains $\alpha_i$ subsequent vectors of the schedule $\mathbf{C}$ that does not contain label $i$.  This implies $\mathbf{C}$ is not a feasible schedule, which contradicts again our assumption. 
	\end{proof}
	
	According to Theorem~\ref{malpha} one can find a schedule for an instance of WSP using a schedule for an instance of PP. A common approach proposed in the literature consists in restricting the search to perfect schedules. We prove next that our heuristic 
	returns a feasible schedule if a perfect schedule exists; in addition, it can also return non-perfect schedules. As we will prove, cases exist when the WSP does not admit a perfect schedule while it does admit a non-perfect one. We will provide such example and show that our heuristic is able to solve it.
	
	The following lemma provides a sufficient condition to exclude existence of a perfect schedule.  An immediate corollary of this lemma and of Theorem~\ref{malpha} is that the heuristic in Theorem~\ref{malpha} can schedule all WSP instances that admit a perfect schedule.
	\begin{lemma}
		\label{mc_lemma_2}
		$ \{\tilde{\alpha}_i\} \notin \text{PP}  \implies \{m_{\mathrm{c}},\{\alpha_i\}\} \notin \text{WSP-perfect}$.
	\end{lemma}
	\begin{proof}
		Assume $\mathbf{C}$ is a perfect schedule for WSP. Then, ${c}_i(t)={c}_j(t^\prime)=k$ implies $i=j$ for all $i,j \in \{1,\ldots,m_{\mathrm{c}}\}$ and $k \in  \{1,\ldots,q\}$. Consider the sequence $\mathbf{C}_{\mathrm{P}}$ as in \eqref{eq:CP1} where $t \geq 1$. Since $\mathbf{C}$ is a perfect schedule, ${c}_i(t)={c}_i(t^\prime)=k$ for every $k \in \{1,\ldots,q\}$. Furthermore, $|t-t^\prime|\leq \alpha_k$ implies ${c_{\mathrm{P}}}(t_1)={c_{\mathrm{P}}}(t_2)=k$ with $t_1=m_{\mathrm{c}}(t-1)+i,~t_2=m_{\mathrm{c}}(t^\prime-1)+i$. Hence, $|t_1-t_2|=m_{\mathrm{c}}|t-t^\prime|\leq m_{\mathrm{c}} \alpha_k$. Consequently, $\mathbf{C}_{\mathrm{P}}$ is a feasible schedule for PP.
	\end{proof}
	
	The following example shows that the converse of Theorem~\ref{malpha} does not hold in general, i.e., $\exists~ \{ m_{\mathrm{c}},\{\alpha_i\} \} \in \text{WSP}$ while $\{\tilde{\alpha}_i\} \notin \text{PP}$. This also indicates the importance of non-perfect schedules. 
	
	\begin{example}[Converse of Theorem~\ref{malpha}] \label{exp:conv2}
		Consider problem instance
		\begin{equation}
		\label{eq:Ex3_instance}
		\{m_{\mathrm{c}},\{\alpha_i\}\}= \{2,\{2,3,3,4,5,5,10\}\}.
		\end{equation}
		While $\{\tilde{\alpha}_i\} \notin \text{PP}$, a schedule with the cyclic part
		\begin{align}
		\mathbf{C}_{\mathrm{r}}=&C_1,C_2,C_3,C_4,C_5,C_6,C_1,C_7,C_8,C_9,\nonumber\\
		&C_5,C_9,C_3,C_{10},C_1,C_6,C_5,C_{11},C_8,C_2~,
		\end{align}
		is feasible for WSP where
		\begin{align}
		C_1&=(1,2),&C_2=(3,4),~&C_3=(1,6),&C_4=(2,5) \nonumber \\
		C_5&=(1,3),&C_6=(4,7),~&C_7=(3,6),&C_8=(1,5)\nonumber \\
		C_9&=(2,4),&C_{10}=(3,5),~&C_{11}=(2,6).&
		\end{align}
	\end{example}
	\begin{remark}
		Since $\{\tilde{\alpha}_i\} \notin \text{PP}$ in Example~\ref{exp:conv2}, $\{m_{\mathrm{c}},\{\alpha_i\}\} \notin \text{WSP-perfect}$ by Lemma ~\ref{mc_lemma_2}. However, $\{m_{\mathrm{c}},\{\alpha_i\}\} \in \text{WSP}$ which provides a negative answer to an open problem in the scheduling community, i.e., whether all feasible instances of WSP admit perfect schedules too, see \cite{bar2007windows}.
	\end{remark}
	
	The next example provides a case in which $ \{\tilde{\alpha}_i\} \in \text{PP} $ while $\{m_{\mathrm{c}},\{\alpha_i\}\} \notin \text{WSP-perfect}$. This implies that the proposed heuristic for WSP can return feasible schedules for instances in which there is no perfect schedule. 
	\begin{example}[] \label{exp:conv3}
		Consider the problem instance
		\begin{equation}
		\label{eq:Ex4_instance}
		\{m_{\mathrm{c}},\{\alpha_i\}\}= \{2,\{2,3,4,5,5,5,7,14\}\}.
		\end{equation}
		In order to find a perfect schedule, one can first compute all possible allocations of agents to one channel or the other, see Table~\ref{Table:bins}. One can verify that $I_1 \notin \text{PP}$ for all allocations in Table~\ref{Table:bins} where $I_1$ is the instance allocated to the first channel. Consequently, $\{m_{\mathrm{c}},\{\alpha_i\}\} \notin \text{WSP-perfect}$.
		\begin{table}[H]
			\caption{Channel allocation of the problem instance~\eqref{eq:Ex4_instance}}
			\label{Table:bins}
			\begin{center}
				\begin{tabular}{ c | c  c } 
					
					& Channel 1 & Channel 2\\
					\hline
					\multicolumn{1}{ c|  } {Allocation 1 }&  $\{2,3,7\}$  & $\{4,5,5,5,14\}$   \\
					\hline
					\multicolumn{1}{ c|  } {Allocation 2} & $\{2,3,14\}$  &  $\{4,5,5,5,7\}$   \\
					\hline
					\multicolumn{1}{ c | } {Allocation 3} & $\{3,5,5,7,14\}$     & $\{2,4,5\}$   \\
					\hline
					\multicolumn{1}{ c | } {Allocation 4} &  $\{2,4,7,14\}$  & $\{3,5,5,5\}$  \\
					\hline
					\multicolumn{1}{ c | } {Allocation 5} &  $\{3,4,5,7\}$ & $\{2,5,5,14\}$ \\
					\hline
					\multicolumn{1}{ c | } {Allocation 6} &  $\{3,4,5,7,14\}$ & $\{2,5,5\}$ \\
					\hline
					\multicolumn{1}{ c | } {Allocation 7} &  $\{3,4,5,5\}$ & $\{2,5,7,14\}$  \\
					\hline
				\end{tabular}
			\end{center}
		\end{table}
		However, a schedule with the following cyclic part is feasible for instance $ \{\tilde{\alpha}_i\}$ of PP
		\begin{align}
		{\mathbf{C}_{\mathrm{P}}}_{\mathrm{r}}=&2,3,4,1,7,6,2,1,5,3,2,1,4,3, \nonumber\\
		& 6,1,2,5,7,1,4,3,2,1,6,8,5,1.
		\end{align}
		This schedule can be transformed into a feasible schedule for WSP with the cyclic part
		\begin{align}
		\mathbf{C}_{\mathrm{r}}=&(2,3),(4,1),(7,6),(2,1),(5,3),(2,1),(4,3), \nonumber \\
		&(6,1),(2,5),(7,1),(4,3),(2,1),(6,8),(5,1).
		\end{align}
	\end{example}
	
	We propose the following algorithm to compute (possibly non-perfect) schedules for WSP.
	\begin{algorithm}
		\begin{algorithmic}[1]
			\footnotesize
			\State Define $\tilde{\alpha}_i$ as in \eqref{new_alpha}
			\If {$\{\tilde{\alpha}_i\} \in \text{PP}$}
			\State find the feasible schedule $\mathbf{C}_{\mathrm{P}}$ for instance $\{\tilde{\alpha}_i\}$ of PP
			\State define $C(t):=\left(c_{\mathrm{P}}(m_{\mathrm{c}}(t-1)+1),\ldots,c_{\mathrm{P}}(m_{\mathrm{c}}t)\right)$
			\State \Return {$\mathbf{C}:= \big{(}C(1),~C(2),\ldots\big{)}$ }
			\Else
			\State \Return{no schedule was found}
			\EndIf
		\end{algorithmic}
		\caption{A heuristic scheduling for WSP (\textbf{input}: $\{\{\alpha_i\}, m_{\mathrm{c}}\}$, \textbf{output}: $\boldsymbol{C}$) }
		\label{alg:WSP}
	\end{algorithm}

	Algorithm~\ref{alg:WSP} checks whether $\{\tilde{\alpha}_i\}$ is accepted by PP or not. Since 
	\begin{itemize}
		\item $\{\tilde{\alpha}_i\} \in \text{PP} \implies \{m_{\mathrm{c}},\{\alpha_i\}\} \in \text{WSP}$,
		\item $ \{\tilde{\alpha}_i\} \notin \text{PP}  \implies \{m_{\mathrm{c}},\{\alpha_i\}\} \notin \text{WSP-perfect}$,
		\item $ \exists~\{m_{\mathrm{c}},\{\alpha_i\}\}:~ \{\tilde{\alpha}_i\} \in \text{PP},~\{m_{\mathrm{c}},\{\alpha_i\}\} \notin \text{WSP-perfect}$,
	\end{itemize}
	Algorithm~\ref{alg:WSP} outperforms the current heuristics in the literature in the sense that it accepts more instances of WSP.

	While $\rho(I) \leq 0.5 m_{\mathrm{c}}$ is a sufficient condition for schedulability of WSP~\cite{bar2003windows}, we provide alternative, less restrictive sufficient conditions in the following theorem.
	\begin{theorem}
		\label{thm:schedulability_thresholdmc2}
		Given an instance $I=\{m_{\mathrm{c}},\{\alpha_i\}\}$ of WSP,
		\begin{enumerate}
			\item if $\rho(I)\leq 0.75m_{\mathrm{c}}$ then $I \in \text{WSP}$,
			\item if $\rho(I)\leq \frac{5}{6} m_{\mathrm{c}}$ and $I$ has only three symbols then $I \in \text{WSP}$.
		\end{enumerate}
	\end{theorem}
	\begin{proof}
		This is a direct consequence of Theorems~\ref{thm:schedulability_threshold} and \ref{malpha}.
	\end{proof}
	
	\section{Main Results: online scheduling}
	\label{section:results_online}
	
	The scheduling techniques, proposed in Section~\ref{section:results_offline} are computed offline solely based on the information available \textit{a priori} and without any online adaptation. The main drawback of this setup is the conservativeness stemming from the fact that the robust invariance condition \eqref{eq:safe_time} must hold for all admissible initial conditions and disturbances. Moreover, packet losses are not explicitly accounted for. This issue has been partially addressed in \cite{ECC}, where an online adaptation of the schedule has been proposed for the single channel case. 
	
	Here we exploit the fact that, differently from offline scheduling, information about the state is available through current or past measurements and can be used to compute less conservative reachable sets in a similar fashion to \eqref{eq:safe_time}. In Section~\ref{Sec_online_subsec_online}, we show that the online scheduling significantly reduces conservativeness. Then, in Section~\ref{subsec:PacketLoss} we extend the results in \cite{ECC} and also provide necessary and sufficient conditions for the existence of a feasible schedule in case of lossy communication link.
	\subsection{Online Scheduling without Packet Losses}
	\label{Sec_online_subsec_online}
	In this subsection we show, under the assumption of no packet losses, how the schedule can be optimized online, based on the available information.
	
	Our strategy is to start with a feasible offline schedule, which we call the \textit{baseline schedule}. Such schedule is then shifted based on estimates of the safe time intervals, which are built upon the current state. In fact, while in equation \eqref{eq:safe_time} the safe time interval is defined as the solution of a reachability problem with $\mathcal{S}_{i,\infty}$ as the initial set, the scheduler may have a better set-valued estimate of the current state of each agent than the whole $\mathcal{S}_{i,\infty}$. This estimate, which we call $\mathcal{O}$, can in general be any set with the following properties, for all $t \geq 0$:
	\begin{subequations}
		\begin{align}
		\label{eq:SchedulerCoDomainDefProperties}
		&\left( x(\tau),\hat{x}(\tau),u(\tau) \right) \in \mathcal{O}(\tau), &\forall \tau \in  \{0,\ldots,t\}, \\
		& \mathcal{O}(t) \subseteq \mathcal{S}_{\infty}, & \text{if}~\delta=1,\\
		&\mathcal{O}(t) \subseteq \Reach{\hat{F}}{1}{\mathcal{O}(t-1)}.
		\end{align}
	\end{subequations}
	
		\begin{example}
			Consider a case in which several automated vehicles are to cross an intersection and the crossing order is communicated to them from the infrastructure, equipped with cameras to measure the states of the vehicles. This corresponds to an SC network, as described by Equation~\eqref{eq:sc_network} in Example~\ref{Exp:case i}, with a scheduler that can measure the state of all agents at all time, but the state measurements have additive noise, i.e., $x_i^s(t)=x_i(t)+w_i(t)$ where $\omega_i \in \mathcal{W}_i$ and $0 \in \mathcal{W}_i$. 
			Then, 
			\begin{equation}
			\mathcal{O}_i(t):=\begin{bmatrix}
			x_i^s(t) \oplus (-\mathcal{W}_i)\\
			(A_i+B_iK_i)^{(t-\tau_i^C(t-1))}x_i^s(\tau_i^C(t-1))\\
			K_i (A_i+B_iK_i)^{(t-\tau_i^C(t-1))}x_i^s(\tau_i^C(t-1))
			\end{bmatrix},
			\end{equation}
			where subscript $i$ refers to agent $i$, and $\tau_i^C$ is the last time when agent $i$ was connected.
			
		\end{example}
	
	Based on set $\mathcal{O}_i(t)$ available at time $t$, we can compute a better estimate of the safe time interval. Let us define this estimate, function of $t$, as follows:
	\begin{equation}
	\gamma^x_i(t):=\max \left\{ t^\prime: \Reach{\hat{F}_i}{t^\prime}{\mathcal{O}_i(t)} \subseteq \MR{i} \right\}.
	\label{eq:GammaOfXDef}
	\end{equation}
	Equations \eqref{eq:safe_time} and \eqref{eq:GammaOfXDef} imply that, for any feasible schedule $\mathbf{C}$,
	\begin{equation}
	\label{eq:alpha_less_than_gamma}
	\gamma_i^x(t)\geq \alpha_i-(t-\tau_i^{\boldsymbol{\mathbf{C}}}(t)),  \qquad \forall i \in  \{1,\ldots,q\},~ \forall \, t>0.
	\end{equation}
	Let us now introduce, for any arbitrarily defined schedule $\mathbf{C}_{\mathrm{o}}$, the quantity
	\begin{equation}
	\gamma^{\mathbf{C}}_i(t):=\min \{t^\prime \ge t:~ i \in C(t^\prime)\}-t,
	\label{eq:GammaDef}
	\end{equation}
	which measures how long agent $i$ will have to wait, at time $t$, before being connected. Using \eqref{eq:GammaDef} and \eqref{eq:GammaOfXDef}, we can formulate a condition for the schedule $\mathbf{C}_{\mathrm{o}}$ to be feasible.
	\begin{definition}[Online feasible schedule]
		\label{pr:gamma_scheduling}
		A schedule~$\mathbf{C}_{\mathrm{o}}$ is online feasible if the \emph{safety residuals} $\mathbf{r}(\mathbf{C}_{\mathrm{o}},t)$ defined as 
		\begin{equation}\label{eq:residuals}
		r_i(\mathbf{C}_{\mathrm{o}},t):=\gamma^x_i(t)-\gamma^{\mathbf{C}_{\mathrm{o}}}_i(t), ~ \forall i \in  \{1,\ldots,q\},
		\end{equation}
		are non-negative, for all $t$, with $\gamma^x(t)$  defined in \eqref{eq:GammaOfXDef} and $\gamma^{\mathbf{C}_{\mathrm{o}}}(t)$ defined in \eqref{eq:GammaDef}.
	\end{definition}
	
	In the job scheduling literature (e.g., \cite{Pinedo08}), the quantities $\gamma^{\mathbf{C}}_i$ correspond to the \textit{completion times} of job $i$, the quantities $\gamma^x_i$ are the \textit{deadlines}, and the quantity $\gamma^{\mathbf{C}}_i-\gamma^x_i = -r_i$ is the \textit{job lateness}.  A schedule for $q$ jobs with deadlines is feasible provided that the maximum lateness is non-positive, that is, all safety residuals are non-negative.
	
	In the following, we formulate an optimization problem to find a recursively feasible online schedule using safety residuals and shifts of the baseline schedule. 
	Given a cycle~$\mathbf{C}_{\mathrm{r}}:=\left( C(1),\ldots,C(T_{\mathrm{r}})\right)$ of the baseline schedule, let
	\begin{equation}
	\label{eq:rotated_schedule}
	R(\mathbf{C}_\mathrm{r},j):=\left( C(j),\ldots,C(T_{\mathrm{r}}),C(1),\ldots,C(j-1) \right)
	\end{equation}
	be a rotation of the sequence $\mathbf{C}_{\mathrm{r}}$ with $j \in \{1,\ldots,T_{\mathrm{r}}\}$. Then, one can compute the shift of the baseline schedule $\mathbf{C}$ which maximizes the minimum safety residual by solving the following optimization problem
	%
	\begin{equation}
	j_t^*:= \arg \max_{j} ~ \min_i \ r_i(R(\mathbf{C}_\mathrm{r},j),t). \label{eq:ResOpt}
	\end{equation}
	The online schedule maximizing the safety residual is then
	\begin{align}
	\label{eq:Optimal_Online_Schedule}
	C^*(t) := C(j_t^*).
	\end{align}
	
	\begin{proposition}
		\label{prop1}
		Assume that the baseline schedule $\mathbf{C}$ is feasible for P\ref{pr:schedulability}. Then, the online schedule $\mathbf{C}^*$ is feasible for P\ref{pr:schedulability}.
	\end{proposition}
	
	\begin{proof}
		At time $t=1$, the baseline schedule is a feasible schedule which implies  $\min_i r_i(R(\mathbf{C}_\mathrm{r},1),1) \geq 0$. As a result, $\min_i r_i(R(\mathbf{C}_\mathrm{r},j_1^*),1) \geq 0$ by construction and schedule $\tilde{\mathbf{C}}$, defined as
		\begin{equation}
		\tilde{\mathbf{C}}:=C(j_1^*),\ldots,C(T_{\mathrm{r}}),\mathbf{C}_{\mathrm{r}},\mathbf{C}_{\mathrm{r}},\ldots~,
		\end{equation}
		is a feasible schedule. Since $C^*(1)=C(j_1^*)$,  
		$$\min_i r_i(R(\mathbf{C}_\mathrm{r},j_1^*+1),2) \geq 0,$$
		which implies 
		$$\min_i r_i(R(\mathbf{C}_\mathrm{r},j_2^*),2) \geq 0.$$
		Consequently,
		\begin{equation}
		\tilde{\mathbf{C}}:=C(j_1^*),C(j_2^*),C(j_2^*+1),\ldots,C(T_{\mathrm{r}}),\mathbf{C}_{\mathrm{r}},\ldots~,
		\end{equation}
		is a feasible schedule. This argument can be used recursively which implies $\mathbf{C}^*$ is a feasible schedule.
	\end{proof}
	\begin{remark}
			The schedule~\eqref{eq:Optimal_Online_Schedule} maximizes the minimum residual, as shown in~\eqref{eq:ResOpt}. That is, the communication is scheduled for the system which is closest to exit~$\mathcal{S}_{\infty}$. Clearly, any function of the residuals could be used. For example, the residuals could be weighted, thus reflecting the priority given to the constraints to be satisfied. \end{remark}
	
	\subsection{Robustness Against Packet Loss}
	\label{subsec:PacketLoss}
	In this subsection, we drop the assumption of no packet losses in the communication link and we consider a communication protocol which has packet delivery acknowledgment. We provide a reconnection strategy to overcome packet losses when the baseline schedule satisfies a necessary condition. Furthermore, we provide the necessary and sufficient conditions for the existence of robust schedules in the presence of packet losses. Then, using these and the results in Section~\ref{Sec_online_subsec_online}, we provide an algorithm to compute an online schedule that is robust to packet losses.. 
	
	Let us consider the binary variable $\nu(t) \in \{0,1\}$, with $\nu(t)=1$ indicating that the packet sent at time $t$ was lost. This binary variable is known to the scheduler if, as we assume, an acknowledgment-based protocol is used for communication.  
	Let us also assume that the maximum number of packets that can be lost in a given amount of time is bounded.
	\begin{assumption} No more than $n_{l,i}$ packets are lost in $\alpha_i$ consecutive steps, i.e.,
		\label{ass:packet_losses}
		\begin{align}
		\sum_{j=t}^{t+\alpha_i-1}\nu(j) \leq n_\mathrm{l,i}, && \forall\ i \in \{1,\ldots,q\},\ \forall \, t \geq 0.
		\label{eq:PacketLoss}
		\end{align}
	\end{assumption}
	Note that \eqref{eq:PacketLoss} defines $q$ different inequalities which must be satisfied at the same time. Additionally, we assume that when a packet is lost, the whole information exchanged at time $t$ is discarded.

	\begin{problem}[P\ref{pr:schedulability_losses}]
		\label{pr:schedulability_losses}
		Given the set of $q$ agents, each described by \eqref{eq:general_system}, an admissible set $\m{A}:= \m{A}_1\times\ldots\times\m{A}_q$, and the set $\boldsymbol{\mathcal{C}}$ of connection patterns \eqref{eq:binary_vectors}, determine if there exists an infinite sequence over the elements of $\boldsymbol{\mathcal{C}}$ such that, 
		\begin{equation}
		z_i(t) \in \MR{i},~\forall z_i(0) \in \MR{i}, \ v_i(t)\in\m{V}_i,~i\in\{1,\ldots,q\},
		\end{equation}
		for $t>0$, provided that Assumption~\ref{ass:packet_losses} holds. 
	\end{problem}
	A schedule solving P\ref{pr:schedulability_losses} is any sequence of $C_j$ such that every agent $i$ is connected at least once every $\alpha_i$ steps in the presence of packet losses satisfying Assumption~\ref{ass:packet_losses}. Instance $I:=\{\boldsymbol{\mathcal{C}},\{\alpha_i\},\{n_\mathrm{l,i} \} \}$ is accepted, i.e.,
	\begin{equation}
	\label{eq:instanceP2}
	I \in \text{P\ref{pr:schedulability_losses}},
	\end{equation}
	if and only if a schedule $\mathbf{C}$ exists that satisfies the scheduling requirements. 
	
	Given a feasible baseline schedule $\mathbf{C}$, we define a \textit{shifted schedule} $\bar{\mathbf{C}}$ as
	\begin{equation}
	\label{eq:schedule_with_losses}
	\bar{C}(t) :=C \left(t-\sum_{j=0}^{t-1} \nu(j)\right),
	\end{equation}
	to compensate the effects of packet losses. We define the maximum time between two successive connections of agent $i$, based on schedule $\mathbf{C}$ as
	\begin{equation}
	T_i := \left( 1+\max_t \ t-\tau_i^{\mathbf{C}}(t)\right),
	\end{equation}
	where the latest connection time $\tau_i^{\mathbf{C}}(t)$ is defined in \eqref{eq:TauLastMeasurement}. Feasibility of the baseline schedule implies $T_i\leq \alpha_i$, for all $i$. 
	
	We prove next that Assumption~\ref{ass:packet_losses} can be used to provide a sufficient condition for the shifted schedule~$\bar{\mathbf{C}}$ to be feasible under packet losses.
	\begin{theorem}[Schedulability under packet losses]
		\label{thm:PacketLoss}
		Let Assumption~\ref{ass:packet_losses} be verified.
		Schedule~$\bar{\mathbf{C}}$ defined in~\eqref{eq:schedule_with_losses} is feasible for $\mathrm{P}\ref{pr:schedulability_losses}$ if and only if
		\begin{equation}
		\alpha_i-T_i\geq n_{l,i},  \qquad  \forall i \in  \{1,\ldots,q\}.
		\end{equation}
	\end{theorem}
	\begin{proof}
		In the error-free schedule $\mathbf{C}$, two consecutive appearances of a symbol $i$ are at most $T_i$ steps apart. In the schedule $\bar{\mathbf{C}}$, during $\alpha_i$ steps at most $n_{l,i}$ retransmissions take place.  Hence, if $\alpha_i-T_i\geq n_{l,i}$, two consecutive occurrences of symbol $i$ are never spaced more than $\alpha_i$ steps, ensuring feasibility of the schedule.
	\end{proof}
	In the sequel, we provide necessary and sufficient conditions for the existence of a baseline schedule which is robust to packet losses. To that end, we define a new set of safe time intervals as
	\begin{equation}
	\big{\{}\beta_i: \beta_i=\alpha_i-n_{l,i},\ \forall i \in  \{1,\ldots,q\} \big{\}}.
	\label{eq:thcon}
	\end{equation}
	\begin{theorem}
		Assume that the communication channel satisfies Assumption~\ref{ass:packet_losses}. Then, 
		\begin{align*}
		\{\boldsymbol{\mathcal{C}},\{\alpha_i\},\{n_\mathrm{l,i} \} \}\in \mathrm{P}\ref{pr:schedulability_losses}&&\Leftrightarrow&&\{\boldsymbol{\mathcal{C}},\{\beta_i\} \}\in \mathrm{P}\ref{pr:schedulability},
		\end{align*}
		with $\beta_i$ defined in \eqref{eq:thcon}.
		\label{thm:PacketLoss2}
	\end{theorem}
	\begin{proof}
		We first prove 
		$$
		\{\boldsymbol{\mathcal{C}},\{\alpha_i\},\{n_\mathrm{l,i} \} \}\in \mathrm{P}\ref{pr:schedulability_losses} \implies \{\boldsymbol{\mathcal{C}},\{\beta_i\} \}\in \mathrm{P}\ref{pr:schedulability}.
		$$ 
		Assume that there exists a feasible schedule for instance $\{\boldsymbol{\mathcal{C}},\{\alpha_i\},\{n_\mathrm{l,i} \} \}$ of $\mathrm{P}\ref{pr:schedulability_losses}$ while it is not feasible for instance $\{\boldsymbol{\mathcal{C}},\{\beta_i\} \}$ of $\mathrm{P}\ref{pr:schedulability}$. This implies
		\begin{equation}
		\exists \, i, t>0: t-\tau_i^{\mathbf{C}}(t) \geq \beta_i+1=\alpha_i-n_{l,i}+1,
		\label{eq:ConExamPro}
		\end{equation}
		where the latest connection time \(\tau_i^{\mathbf{C}}(t)\) is defined in \eqref{eq:TauLastMeasurement}. Assume \(n_{l,i}\) consecutive packets are lost starting from time \(t+1\), such that $\tau_i^{\mathbf{C}}(t+n_{l,i})=\tau_i^{\mathbf{C}}(t)$. This implies 
		\begin{equation}
		\left( t+n_{l,i} \right)-\tau_i^{\mathbf{C}}(t+n_{l,i}) \geq \alpha_i+1,
		\end{equation}
		such that agent $i$ did not receive any packet for $\alpha_i$ consecutive steps, i.e., $\{\boldsymbol{\mathcal{C}},\{\alpha_i\},\{n_\mathrm{l,i} \} \}\notin \mathrm{P}\ref{pr:schedulability_losses}$.
		
		In order to prove
		$$
		\{\boldsymbol{\mathcal{C}},\{\beta_i\} \}\in \mathrm{P}\ref{pr:schedulability} \implies \{\boldsymbol{\mathcal{C}},\{\alpha_i\},\{n_\mathrm{l,i} \} \}\in \mathrm{P}\ref{pr:schedulability_losses},
		$$ 
		consider feasible schedule $\mathbf{C}$ for instance $\{\boldsymbol{\mathcal{C}},\{\beta_i\} \}$ of $\mathrm{P}\ref{pr:schedulability}$. Each packet loss causes one time step delay in receiving the measurement for agent $i$, see \eqref{eq:schedule_with_losses}, and since these packet losses can at most cause \(n_{l,i}\) time step delays between two connection times, agent $i$ would be connected at least once during each \(\beta_i+n_{l,i} = \alpha_i\) time steps. This implies $\bar{\mathbf{C}}$ defined in \eqref{eq:schedule_with_losses} is a feasible schedule for instance $\{\boldsymbol{\mathcal{C}},\{\alpha_i\},\{n_\mathrm{l,i} \} \}$ of $\mathrm{P}\ref{pr:schedulability_losses}$.
	\end{proof}
	
	Theorem~\ref{thm:PacketLoss2} implies that $\mathrm{P}\ref{pr:schedulability_losses}$ can be cast in the framework of $\mathrm{P}\ref{pr:schedulability}$ by using equation~\eqref{eq:thcon} to define $\{\beta_i\}$ based on $\{\alpha_i\}$ and $\{n_\mathrm{l,i}\}$.

	Since the shifted schedule $\bar{\mathbf{C}}$ provides a feasible robust schedule against packet losses, one can use the online scheduling method proposed in the previous subsection to improve safety of this robust schedule. To that end, we define the number of packet losses that can occur before agent $i$ receives a measurement from time $t$ as
	\begin{equation}
	n_i^{\mathbf{C}}(t):=\{\min_{t^\prime,n} n: t^\prime \geq t,~i \in C(t^\prime),~\sum_{j=t}^{t^\prime}\nu(j)\leq n  \}.
	\label{eq:packet_loss_remained}
	\end{equation}
	\begin{definition}[Robust online feasible schedule]
		\label{pr:gamma_scheduling2}
		A schedule~${\mathbf{C}}$ is robust online feasible if the \emph{robust safety residuals} $\bar{\mathbf{r}}(\mathbf{C},t)$ defined as 
		\begin{equation}\label{eq:residuals2}
		\bar{r}_i(\mathbf{C},t):=\gamma^x_i(t)-\gamma^{\mathbf{C}}_i(t)-n_i^{\mathbf{C}}(t),~\forall i \in \{1,\ldots,q\},
		\end{equation}
		are non-negative for all $t$, with $\gamma^x(t)$  defined in \eqref{eq:GammaOfXDef}, $\gamma^{\mathbf{C}}(t)$ defined in \eqref{eq:GammaDef}, and $n_i^{\mathbf{C}}(t)$ defined in \eqref{eq:packet_loss_remained}.
	\end{definition}
	Similarly to the case without packet losses, the following optimization problem maximizes the minimum robust safety residuals for a lossy channel.
	\begin{equation}
	\bar{j}_t^*:= \arg \max_{j} ~ \min_i \ \bar{r}_i(R(\mathbf{C}_\mathrm{r},j),t).
	\label{eq:ResOpt2}
	\end{equation}
	The online schedule is then given by
	\begin{align}
	\label{eq:Optimal_Online_Schedule_losses}
	\bar{C}^*(t) := C(\bar{j}_t^*).
	\end{align}
	\begin{proposition} \label{prop2}
		Assume that the baseline schedule $\bar{\mathbf{C}}$ is feasible for P\ref{pr:schedulability_losses}. Then, the online schedule $\bar{C}^*(t)$ is feasible for all $t$.
	\end{proposition}
	\begin{proof}
		This can be proved similarly to Proposition~\ref{prop1} and Theorem~\ref{thm:PacketLoss}.
	\end{proof}
	The following algorithm, based on Theorem~\ref{thm:PacketLoss2}, returns an online robust feasible schedule for P\ref{pr:schedulability_losses}.
	\begin{algorithm}
		\begin{algorithmic}[1]
			\footnotesize
			\State Define $\beta_i$ using \eqref{eq:thcon}
			\If {$\{\beta_i\} \in \text{PP}$}
			\State find a schedule $\mathbf{C}_{\mathrm{P}}$ for instance $\{\beta_i\}$ of PP
			\State define $C(t):=C_j$ when $c_{\mathrm{P}}(t)=j,~\forall j$
			\ForAll{t}
			\State find $\bar{C}^*(t)$ as in $\eqref{eq:Optimal_Online_Schedule_losses}$
			\EndFor
			\State \Return $\bar{\mathbf{C}}^*$ 
			\Else
			\State \Return {no schedule was found}
			\EndIf
		\end{algorithmic}
		\caption{Robust online scheduling for P\ref{pr:schedulability} (\textbf{input}: ($\{\alpha_i\},\{(n_{l,i},T_i)\}, \ldots $) , \textbf{output}: $\bar{\mathbf{C}}^*$)}
		\label{alg:robust_online}
	\end{algorithm}
	\section{Numerical results}
	\label{section:numerics}
	We now discuss some numerical examples in order to illustrate and evaluate the effectiveness of the proposed methods. First, we evaluate Algorithms~\ref{alg:General_P1}~and~\ref{alg:WSP}. Then, we provide a trajectory tracking scenario for remotely controlled vehicles with limited number of lossy communication channels. 
	\subsection{Evaluation of Algorithm~\ref{alg:General_P1}}
	We have considered 1000 networks with a random number of agents (2 to 4), random safe time intervals (2 to 8), and random connection patterns (1 to 4). These random instances are used for evaluating the three following implementations:
	\begin{itemize} 
		\item $M_1$: solve optimization problem~\eqref{eq:qnary_formulation_exact_solution};
		\item $M_2^{\mathrm{A}1}$: use Algorithm~\ref{alg:General_P1} in combination with optimization problem~\eqref{eq:qnary_formulation_exact_solution} to solve PP;
		\item $M_3^{\mathrm{A}1}$:  use Algorithm~\ref{alg:General_P1} in combination with the double-integer method proposed in \cite{Chan92a} to solve PP.
	\end{itemize}
	
	We have used Gurobi to solve the integer problems and provided a summary of the results in Table~\ref{table:comparison}. $M_1$ is exact and returns false positives nor false negatives. Although $M_2^{\mathrm{A}1}$ and $M_3^{\mathrm{A}1}$ do not return any false positive, they might return false negatives. Furthermore, a false negative answer in $M_2^{\mathrm{A}1}$ implies the same for $M_3^{\mathrm{A}1}$ since the latter uses a heuristic to solve PP while the former finds a schedule for PP whenever it exists. Note also that $M_2^{\mathrm{A}1}$ and $M_3^{\mathrm{A}1}$ do not necessarily return a solution with the minimum period length. 
	
	\begin{table}[H]
		\caption{Comparison of $M_1$, $M_2^{\mathrm{A}1}$,~ and $M_3^{\mathrm{A}1}$}
		\label{table:comparison}
		\begin{center}
			\begin{tabular}{ c| c  c  c } 
				& $M_1$ & $M_2^{\mathrm{A}1}$ & $M_3^{\mathrm{A}1}$\\
				\hline
				\multicolumn{1}{ c|  } {True Positive }&  960 & 958 & 958  \\
				\hline
				\multicolumn{1}{ c|  } {False Positive} & 0 &  0 & 0 \\
				\hline
				\multicolumn{1}{ c | } {True Negative} & 40   & 40  & 40 \\
				\hline
				\multicolumn{1}{ c | } {False Negative} &  0 & 2 & 2 \\
				\hline
				
			\end{tabular}
		\end{center}
	\end{table}
	
	In Table~\ref{table:comparison3} and Fig.~\ref{fig:AlgComparisons1}, we tested $M_1$, $M_2^{\mathrm{A}1}$ and $M_3^{\mathrm{A}1}$ on larger randomly generated networks (2 to 11 agents, 2 to 21 safe time intervals, 1 to 11 random connection patterns). To limit the computation-time, we had to halt the execution of $M_1$ and $M_2^{\mathrm{A}1}$ when no schedule of period $ \leq 70$ was found.  We labeled \textit{undecided} the instances for which these two methods were halted.
	\begin{table}[H]
		\caption{Comparison of $M_1$, $M_2^{\mathrm{A}1}$,~ and $M_3^{\mathrm{A}1}$}
		\label{table:comparison3}
		\begin{center}
			\begin{tabular}{c |c c c} 
				
				& $M_1$ & $M_2^{\mathrm{A}1}$ & $M_3^{\mathrm{A}1}$\\
				\hline
				\multicolumn{1}{ c|  } {accepted instances} & 887  & 872  & 853    \\
				
				\multicolumn{1}{ c|  } {average time (sec)} & 1.7915 & 1.3119  & 0.5143 \\
				\hline
				\multicolumn{1}{ c|  } {undecided instances} &  102 &   32 &  0 \\
				
				\multicolumn{1}{ c|  } {average time (sec)} & 61.0570 & 48.9156 & 0 \\
				\hline
				\multicolumn{1}{ c|  } {rejected instances} & 11   &  96 &  147 \\
				
				\multicolumn{1}{ c|  } {average time (sec)} & 53.7730 &  1.5442 & 0.5333 \\
				\hline
			\end{tabular}
		\end{center}
	\end{table}

	\begin{figure}
		\centering
		\begin{subfigure}[b]{0.24\textwidth}
			\centering
			\includegraphics[width=\textwidth]{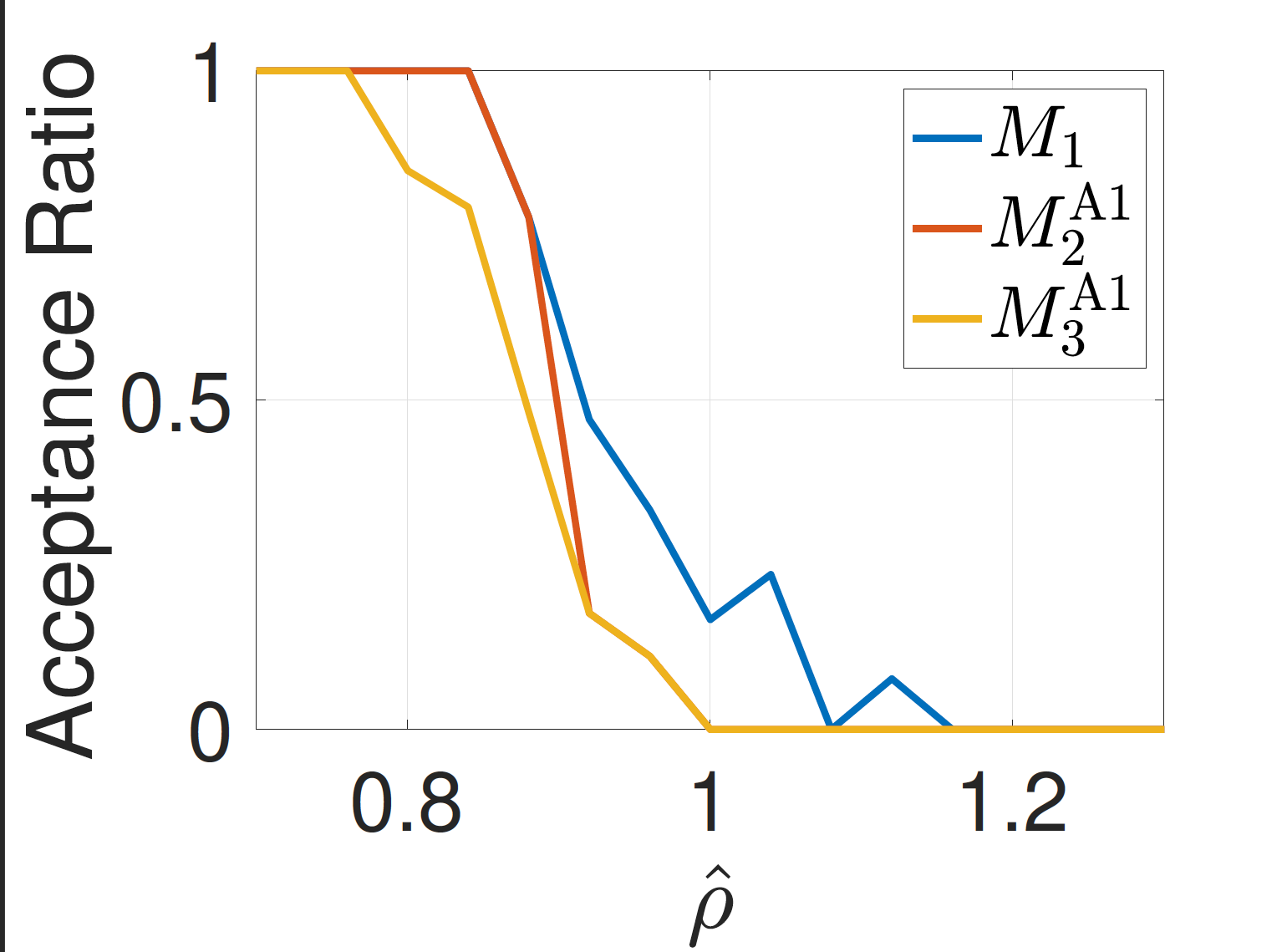}
			\caption{}
			\label{fig:M4M5Mrho}
		\end{subfigure}
		\hfill
		\begin{subfigure}[b]{0.24\textwidth}
			\centering
			\includegraphics[width=\textwidth]{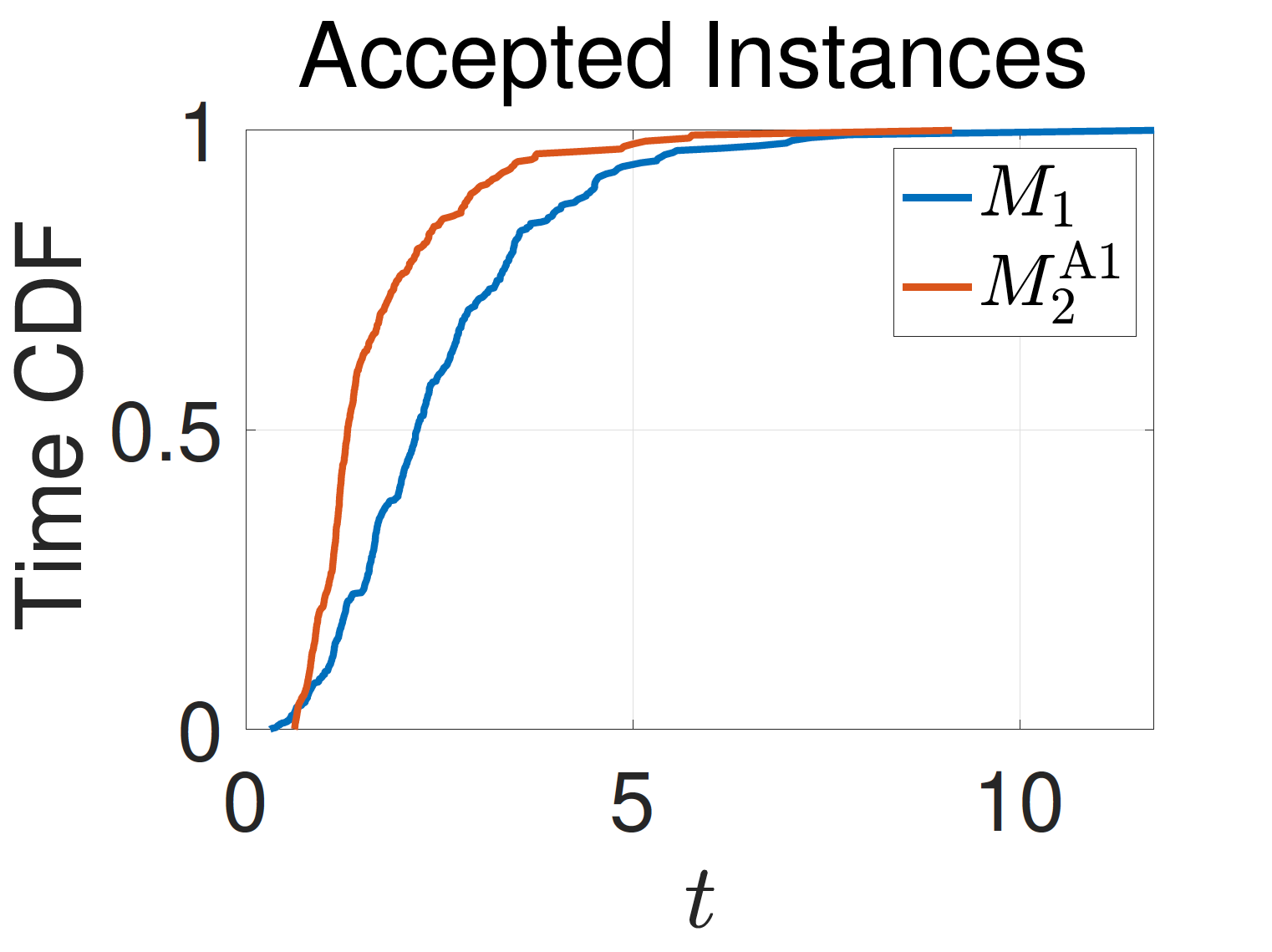}
			\caption{}
			\label{fig:M1M2_times_accepted}
		\end{subfigure}
		\begin{subfigure}[b]{0.24\textwidth}
			\centering
			\includegraphics[width=\textwidth]{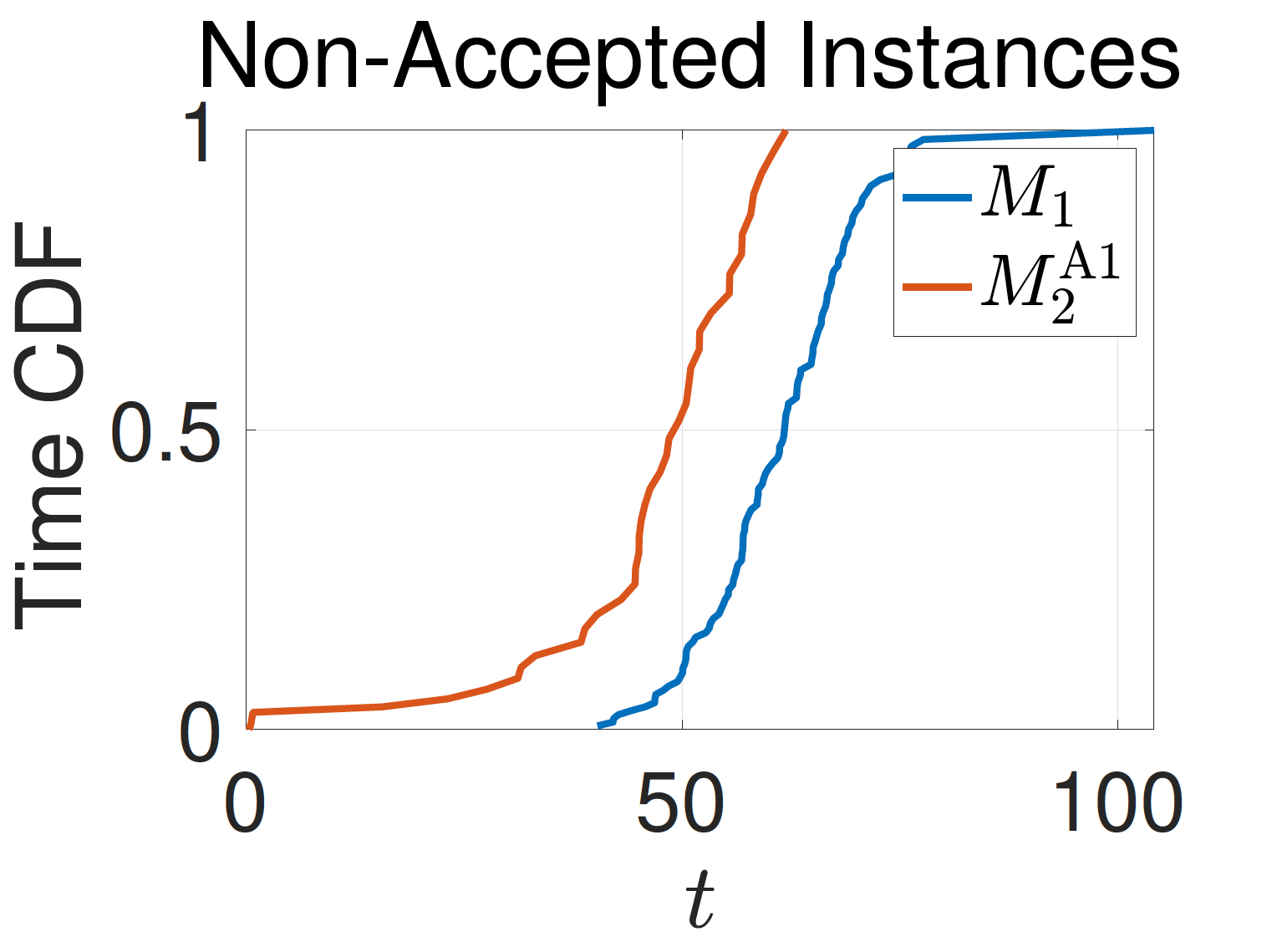}
			\caption{}
			\label{fig:M1M2_times_rejected}
		\end{subfigure}
		\hfill
		\begin{subfigure}[b]{0.24\textwidth}
			\centering
			\includegraphics[width=\textwidth]{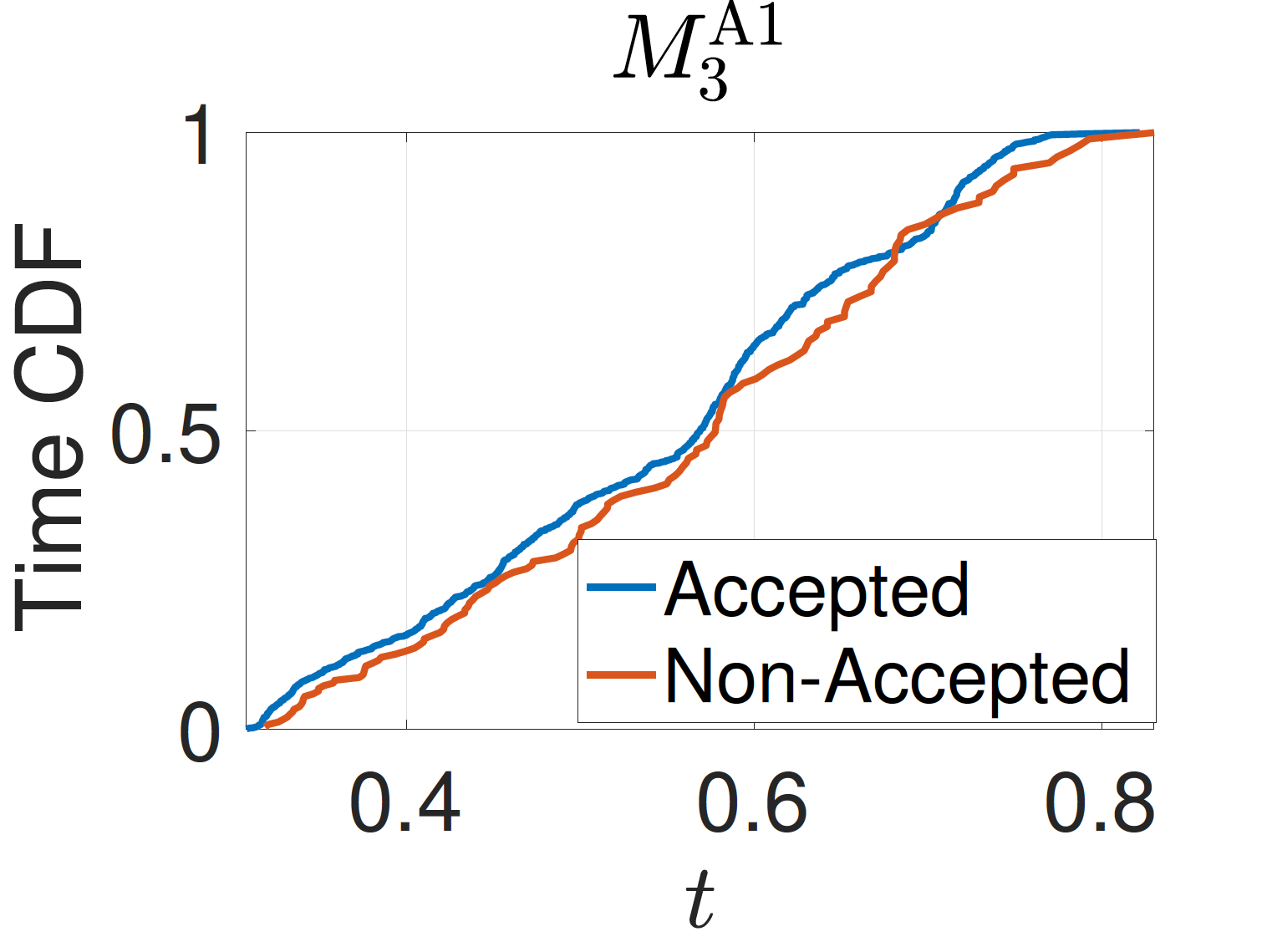}
			\caption{}
			\label{fig:M3_times_accepted_rejected}
		\end{subfigure}
		\caption
		{Comparison of acceptance ratio, with respect to the assigned density function $\hat{\rho}$, and the Cumulative Distribution Function (CDF) of the computation time for $M_1$, $M_2^{\mathrm{A}1}$, and $M_3^{\mathrm{A}1}$ methods}
		\label{fig:AlgComparisons1}
	\end{figure}
	Although $M_3^{\mathrm{A}1}$ might result in a few false negatives, Table~\ref{table:comparison3} indicates that its average computation time is significantly lower than the corresponding average computation times of $M_1$ and $M_2^{\mathrm{A}1}$. More importantly, Fig.~\ref{fig:M1M2_times_accepted}, \ref{fig:M1M2_times_rejected}, and \ref{fig:M3_times_accepted_rejected} demonstrate that $M_3^{\mathrm{A}1}$ has a lower computation time for almost all considered instances. Note that $M_1$ accepts some instances with an assigned density $\hat{\rho}$ greater than one, as shown in Fig.~\ref{fig:M4M5Mrho}, which implies that converse of Theorem~\ref{thm:schedulability_thresholdmcS} does not hold.
	\subsection{Evaluation of Algorithm~\ref{alg:WSP}}
	In this subsection we evaluate the proposed heuristic in Algorithm~\ref{alg:WSP} to find a feasible schedule for P\ref{pr:schedulability} when $\boldsymbol{\mathcal{C}}$ is defined as in \eqref{eq:mc_channels}. We have generated 1000 networks with five agents, random safe time intervals (from 2 to 7), minimum number of channels required for schedulability, i.e., $m_{\mathrm{c}}=\lceil \sum_i \frac{1}{\alpha_i} \rceil $. These random instances are used for evaluating the three following implementations:
	\begin{itemize} 
		\item $M_1$: solve optimization problem~\eqref{eq:qnary_formulation_exact_solution};
		\item $M_2^{\mathrm{A}2}$: use Algorithm~\ref{alg:WSP} in combination with optimization problem~\eqref{eq:qnary_formulation_exact_solution} to solve PP;
		\item $M_3^{\mathrm{A}2}$:  use Algorithm~\ref{alg:WSP} in combination with optimization problem~\eqref{eq:qnary_formulation_exact_solution} to solve PP;
	\end{itemize}
	The simulation results are provided in Table~\ref{table:comparison4}.
	\begin{table}[H]
		\caption{Comparison of $M_1$, $M_2^{\mathrm{A}2}$,~ and $M_3^{\mathrm{A}2}$}
		\label{table:comparison4}
		\begin{center}
			\begin{tabular}{ c | c  c  c } 
				
				& $M_1$ & $M_2^{\mathrm{A}2}$ & $M_3^{\mathrm{A}2}$\\
				\hline
				\multicolumn{1}{ c|  } {True Positive }&  994 & 994 & 987  \\
				\hline
				\multicolumn{1}{ c|  } {False Positive} & 0 &  0 & 0 \\
				\hline
				\multicolumn{1}{ c | } {True Negative} & 6   & 6  & 6\\
				\hline
				\multicolumn{1}{ c | } {False Negative} &  0 & 0 & 7 \\
				\hline
			\end{tabular}
		\end{center}
	\end{table}
	
	Once again, the three implementations were tested on larger instances by halting $M_1$ and $M_2^{\mathrm{A}2}$ if no schedule of length $\leq 70$ was found.  The results are  reported in Table~\ref{table:comparison5} and Fig.~\ref{fig:AlgComparisons2}. Define a \textit{normalized density function} as
	\begin{equation}
	\tilde{\rho}:=\frac{1}{m_{\mathrm{c}}} \left(\sum_i \frac{1}{\alpha_i}\right),
	\end{equation}
	for sake of a meaningful comparison.
	
	\begin{table}[H]
		\caption{Comparison of $M_1$, $M_2^{\mathrm{A}2}$,~ and $M_3^{\mathrm{A}2}$}
		\label{table:comparison5}
		\begin{center}
			\begin{tabular}{c |c c c} 
				
				& $M_1$ & $M_2^{\mathrm{A}2}$ & $M_3^{\mathrm{A}2}$\\
				\hline
				\multicolumn{1}{ c|  } {accepted instances} & 984  & 980  & 909    \\
				
				\multicolumn{1}{ c|  } {average time (sec)} & 5.0353 & 5.5204  & 1.3043e-04 \\
				\hline
				\multicolumn{1}{ c|  } {undecided instances} &  16 & 20   &  0 \\
				
				\multicolumn{1}{ c|  } {average time (sec)} &  267.1501  & 139.8464  & 0 \\
				\hline
				\multicolumn{1}{ c|  } {rejected instances} & 0   & 0  &  91 \\
				
				\multicolumn{1}{ c|  } {average time (sec)} & 0 & 0  & 1.0384e-04 \\
				\hline
			\end{tabular}
		\end{center}
	\end{table}

	\begin{figure}
	\centering
	\begin{subfigure}[b]{0.24\textwidth}
		\centering
		\includegraphics[width=\textwidth]{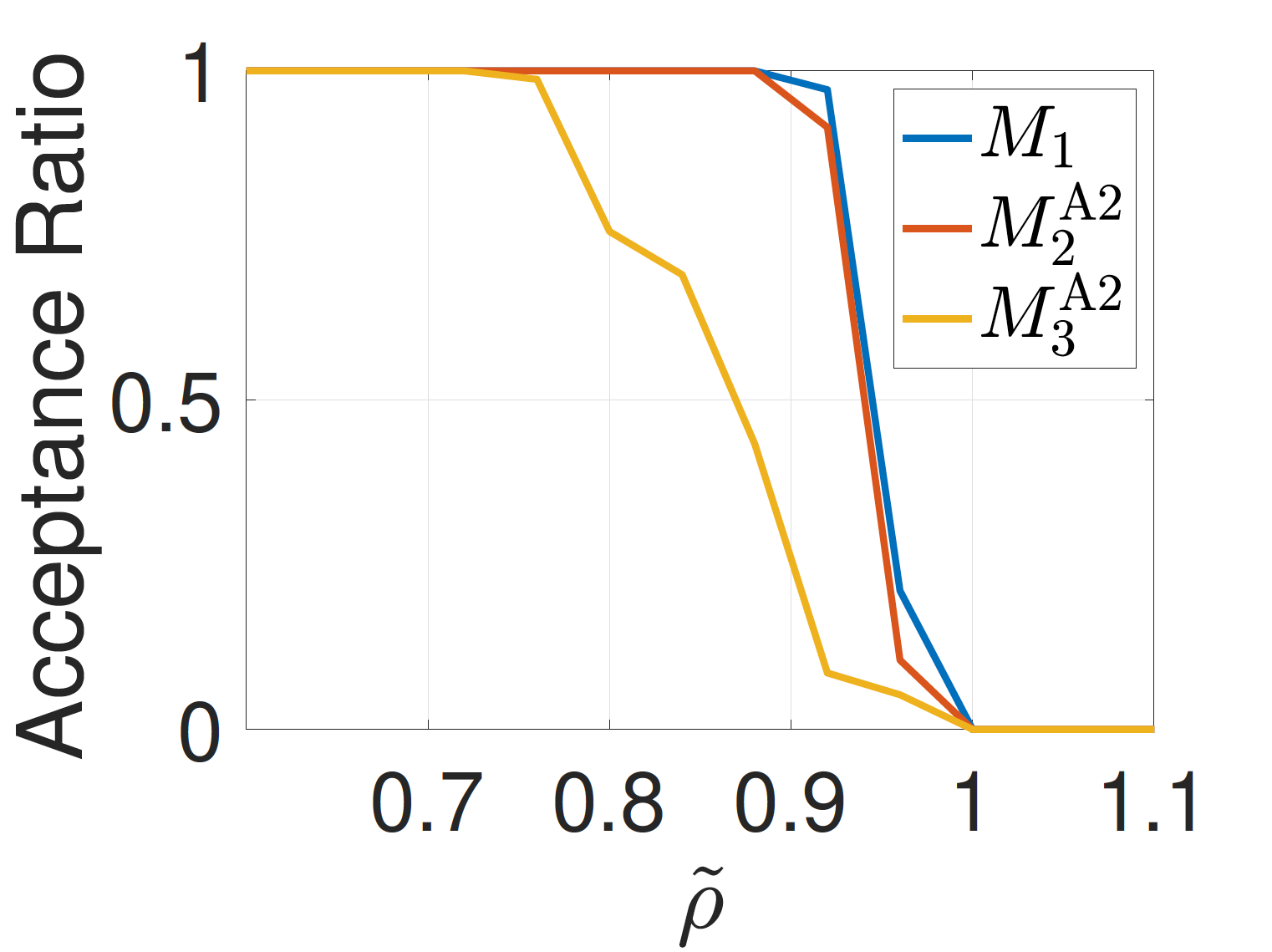}
		\caption{}
			\label{fig:M4M5M6rho}
	\end{subfigure}
	\hfill
	\begin{subfigure}[b]{0.24\textwidth}
		\centering
		\includegraphics[width=\textwidth]{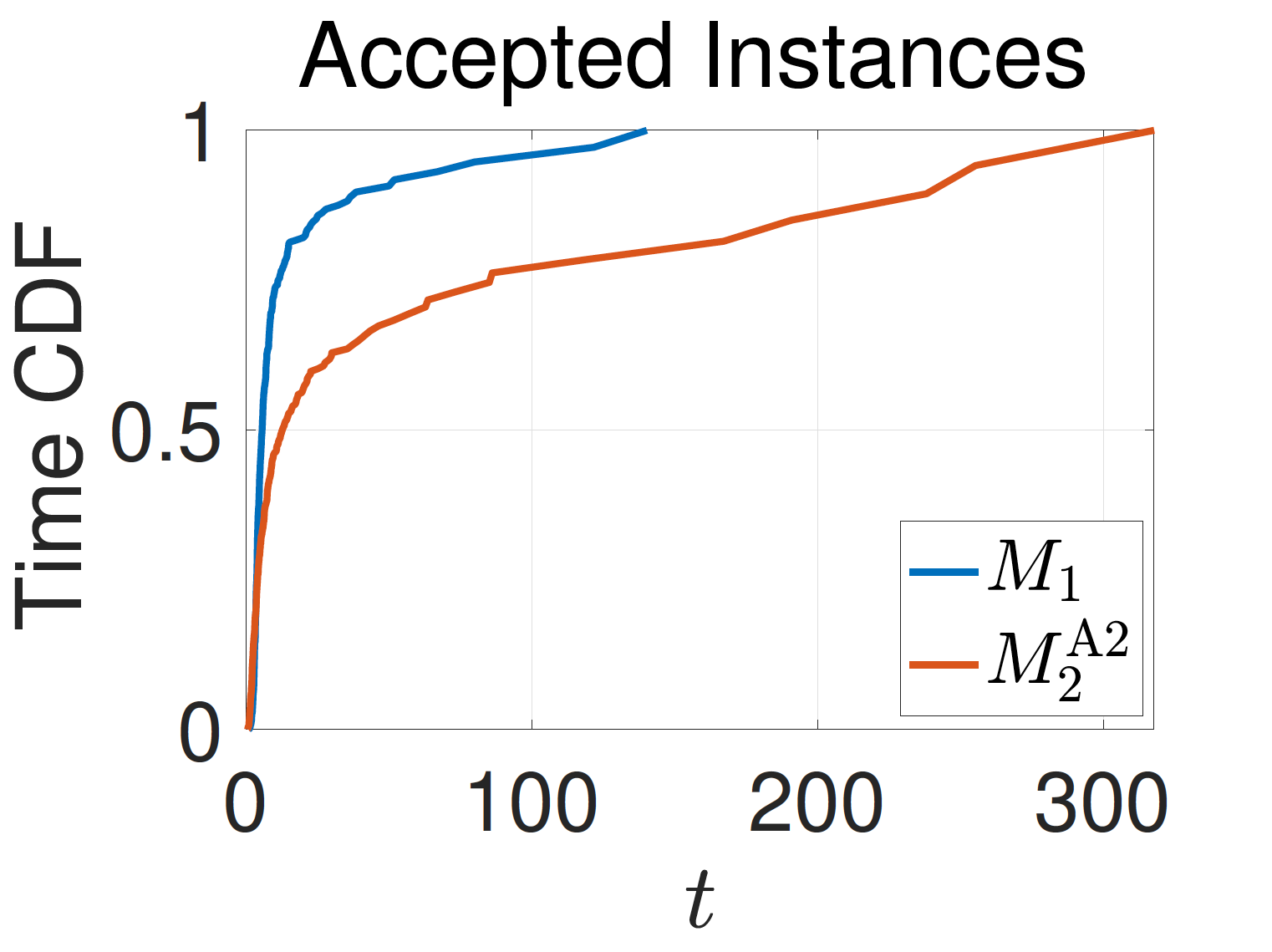}
		\caption{}
			\label{fig:M4M5_times_accepted}
	\end{subfigure}
	\begin{subfigure}[b]{0.24\textwidth}
		\centering
		\includegraphics[width=\textwidth]{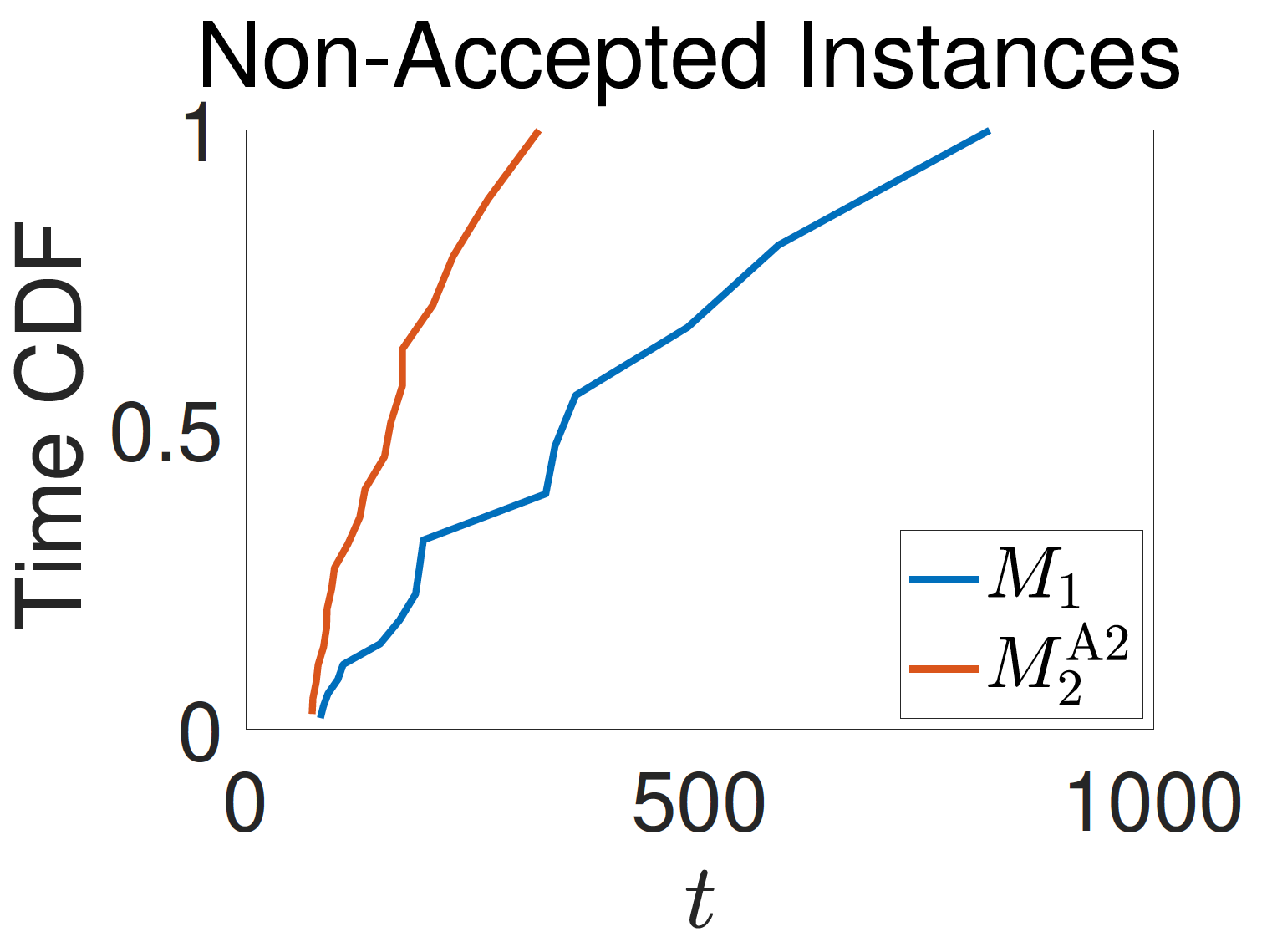}
		\caption{}
			\label{fig:M4M5_times_rejected}
	\end{subfigure}
	\hfill
	\begin{subfigure}[b]{0.24\textwidth}
		\centering
		\includegraphics[width=\textwidth]{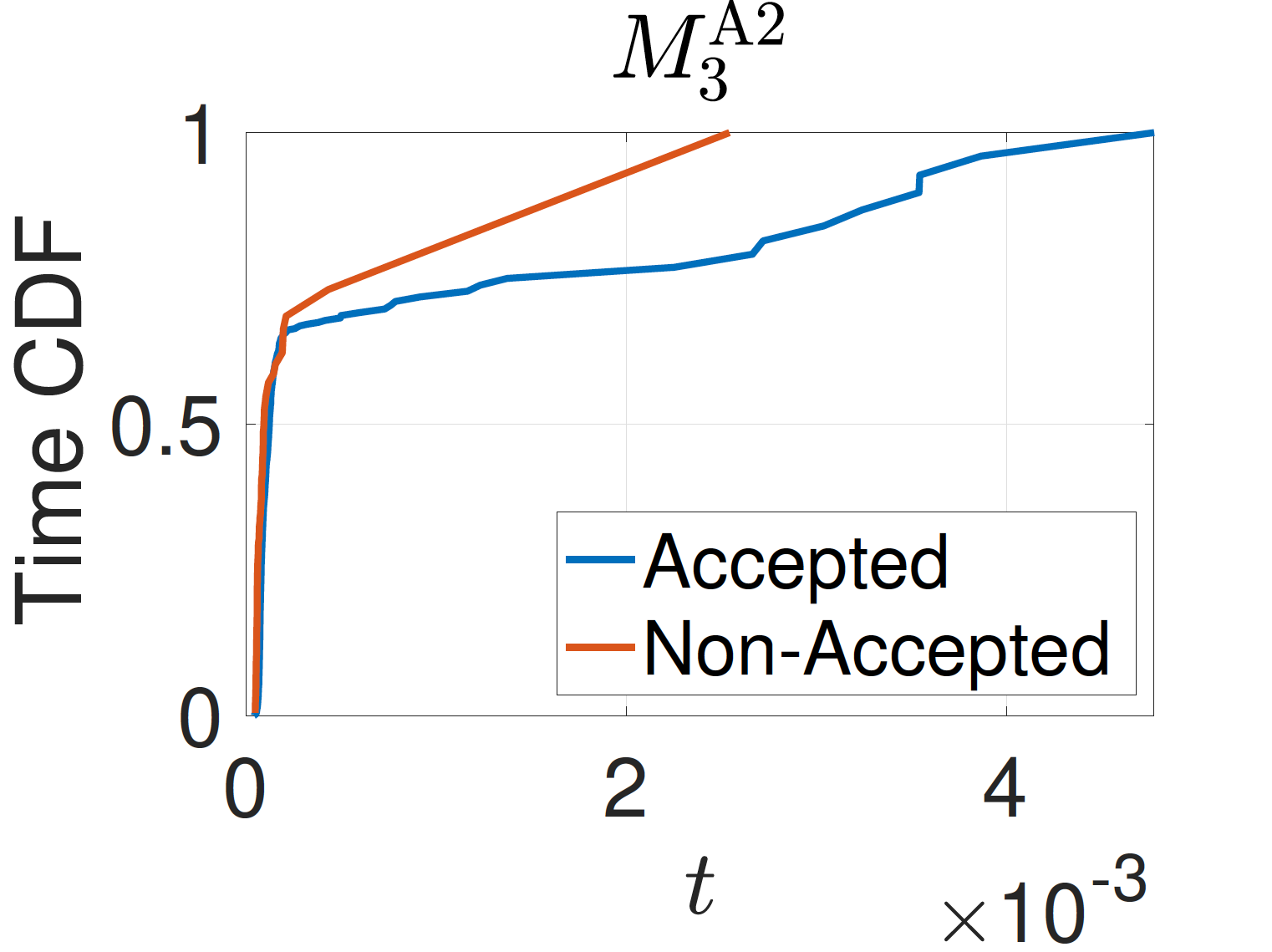}
		\caption{}
			\label{fig:M6_times_accepted_rejected}
	\end{subfigure}
	\caption
	{
		Comparison of acceptance ratio, with respect to the normalized density function $\tilde{\rho}$, and the Cumulative Distribution Function (CDF) of the computation time for $M_1$, $M_2^{\mathrm{A}2}$, and $M_3^{\mathrm{A}2}$ methods
	}
	\label{fig:AlgComparisons2}
\end{figure}
	Although $M_3^{\mathrm{A}2}$ might result in a few false negatives, Table~\ref{table:comparison5} indicates that its average computation time is drastically lower than the corresponding average computation times of $M_1$ and $M_2^{\mathrm{A}2}$. More importantly, Fig.~\ref{fig:M4M5_times_accepted}, \ref{fig:M4M5_times_rejected}, and \ref{fig:M6_times_accepted_rejected} demonstrate that $M_3^{\mathrm{A}1}$ has a lower computation time for almost all considered instances. Note that the simulations confirm the result of Theorem~\ref{thm:schedulability_thresholdmc2}: as displayed in Fig.~\ref{fig:M4M5M6rho} any instance with $\tilde{\rho}=\frac{\rho}{m_{\mathrm{c}}} \leq 0.75$ is schedulable.
	\subsection{Remotely Controlled Vehicles}
	In this subsection, two numerical examples are given to illustrate the introduced concepts and algorithms. First we consider a tracking problem for vehicles with performance/safety constraints on the errors; this problem can be translated to P\ref{pr:schedulability} (see \cite{colombo2018invariance}). Then, we consider a tracking problem where the communication network is subject to packet losses.
	
	\begin{example}[Networked control vehicles without packet loss]\label{ex:ex_tracking}
		
		Consider a case of eight remotely controlled vehicles, described by the models
		\begin{align}\label{eq:vehicles}
		x_i(t+1) &= A_ix_i(t)+B_iu_i(t)+E_iv_i(t),\Forall i
		\end{align}
		where~
		\begin{equation}
		A_i=\left[
		\begin{array}{ccc}
		1  & h  & 0  \\
		0  & 1  & h  \\
		0  &  0 & 1-\frac{h}{\tau_i}  
		\end{array}
		\right],~B_i=\left[
		\begin{array}{c}
		0\\
		0\\
		\frac{h}{\tau_i}
		\end{array}
		\right],~E_i=\left[
		\begin{array}{c}
		0\\
		0\\
		1
		\end{array}
		\right],
		\end{equation}
		and~$\tau_1=0.1$,~$\tau_2=0.2$,~$\tau_3=0.3$,~$\tau_4=0.4$,~$\tau_5=\ldots=\tau_8=0.5$, and~$h=0.2$. 
		
		The longitudinal motion of these vehicles must track their reference trajectories within prescribed error bounds, to realize a specified traffic scenario.  Such situations occur, for instance, when setting up full-scale test scenarios for driver-assist systems. The reference state trajectories are generated by
		\begin{equation}
		x_i^d(t+1)=A_i x_i^d(t)+B_iu_i^d(t),~\forall i,
		\end{equation}
		while the tracking inputs are defined as
		\begin{equation}
		u_i(t):=u_i^d(t)+\tilde{u}_i(t).
		\end{equation}
		The error dynamics for each vehicle is
		\begin{align}\label{eq:vehicles2}
		e_i(t+1) &= A_ie_i(t)+B_i \tilde{u}_i(t)+E_iv_i(t),\Forall i
		\end{align}
		where $e_i(t):=x_i(t)-x_i^d(t)$ is the difference between the state and the desired state and $\tilde{u}_i(t)=u_i(t)-u_i^d(t)$ is the difference between the system input and the input's feed-forward term. 
		We assume that the controller is always connected with the actuator (SC network), i.e., $\delta_u=1$ in Fig.~\ref{fig:system_1}, while the sensor is connected to the controller through a network, i.e., $\delta_s(t)=C(t)$ in Fig.~\ref{fig:system_1}. 
		We consider the feedback terms as $\tilde{u}_i(t)=-K_i\hat{e}_i(t)$ where $\hat{e}_i(t)$ is the tracking error estimation and it is specified by
		\begin{equation}
		\hat{e}_i(t)=\begin{cases}
		e_i(t), & \text{if } i \in C(t)\\
		A_i \hat{e}_i(t-1)+B_i \tilde{u}(t-1), & \text{if } i \notin C(t)
		\end{cases}.
		\end{equation}
		Feedback gains $K_i$ are calculated by solving LQR problems with cost gains \(Q=\text{diag}([10,1,0.1]),\ R=0.1\). Furthermore, $\mathcal{U}_i=\{-10\le \tilde{u}_i \le 10\},$ and $\mathcal{V}_i=\{ |v_i| \le \tilde{v}_i\}$ with \(\tilde{v}_1=3.4,\ \tilde{v}_2=2.1,\ \tilde{v}_3=1.1,\ \tilde{v}_4=0.6,\ \tilde{v}_5=\ldots=\tilde{v}_8=0.4\) are the set of admissible control inputs and the bound on the disturbances, respectively.
		
		For each system, the admissible tracking errors belong to the set
		\begin{equation}
		\mathcal{E}_i=\left\{ e_i \in \mathbb{R}^3:  \left[ 
		\begin{array}{c}
		-1\\
		-5\\
		-10
		\end{array}
		\right] \le 
		e_i \le \left[
		\begin{array}{c}
		1\\
		5\\
		10
		\end{array}
		\right]\right\},~\Forall i.
		\end{equation}
		In this example, safe time intervals are \(\alpha_1=2,~\alpha_2=3,~\alpha_3=4,~\alpha_4=5,~\alpha_5=\ldots=\alpha_8=6\) using  \eqref{eq:safe_time}. Assuming there are two communication channels, i.e., \(m_{\mathrm{c}}=2\), finding a feasible schedule is not straightforward. Nevertheless, one can use Algorithm~\ref{alg:WSP} to find a feasible schedule; the cyclic part of such a schedule is as follows
		\begin{equation}
		\mathbf{C}_{\mathrm{r}}=C_1,C_2,C_3,C_4,C_5,C_6,C_1,C_7,C_3,C_8,C_5,C_9~,
		\end{equation}
		where
		\begin{align}
		C_1&=(1,2),&C_2=(3,5),~&C_3=(1,6),&C_4=(4,2) \nonumber \\
		C_5&=(1,7),&C_6=(3,8),~&C_7=(4,5),&C_8=(3,2) \nonumber  \\
		C_9&=(4,8).
		\end{align}
		The tracking errors for the above schedule, along with the corresponding feedback control actions, are reported in Fig.~\ref{fig:vehicle_errors}, which shows them in their admissible sets. Note that in this example, the scheduler is designed offline.

	\begin{figure}
	\centering
	\begin{subfigure}[b]{0.24\textwidth}
		\centering
		\includegraphics[width=\textwidth]{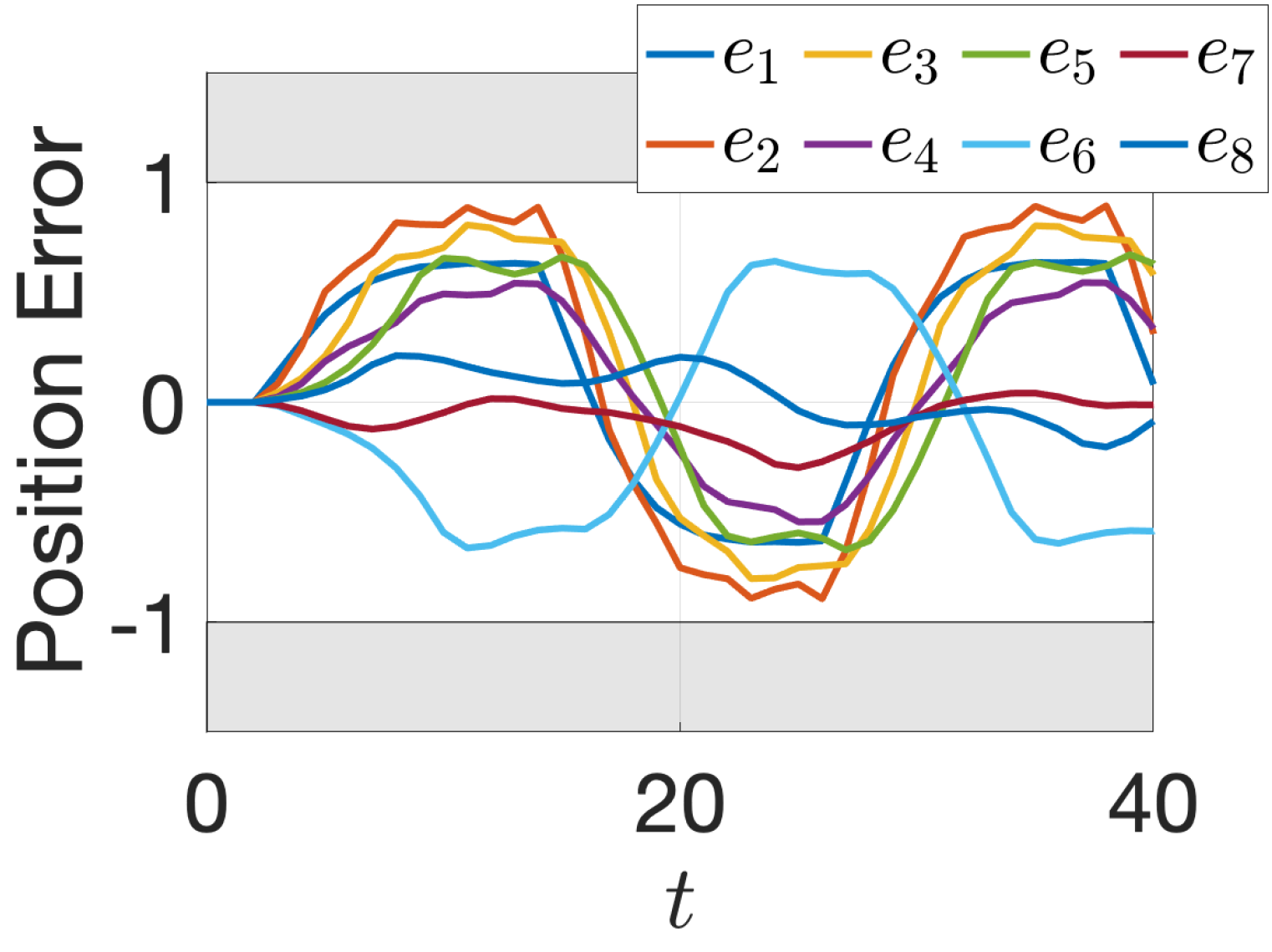}
		\caption{}
	\end{subfigure}
	\hfill
	\begin{subfigure}[b]{0.24\textwidth}
		\centering
		\includegraphics[width=\textwidth]{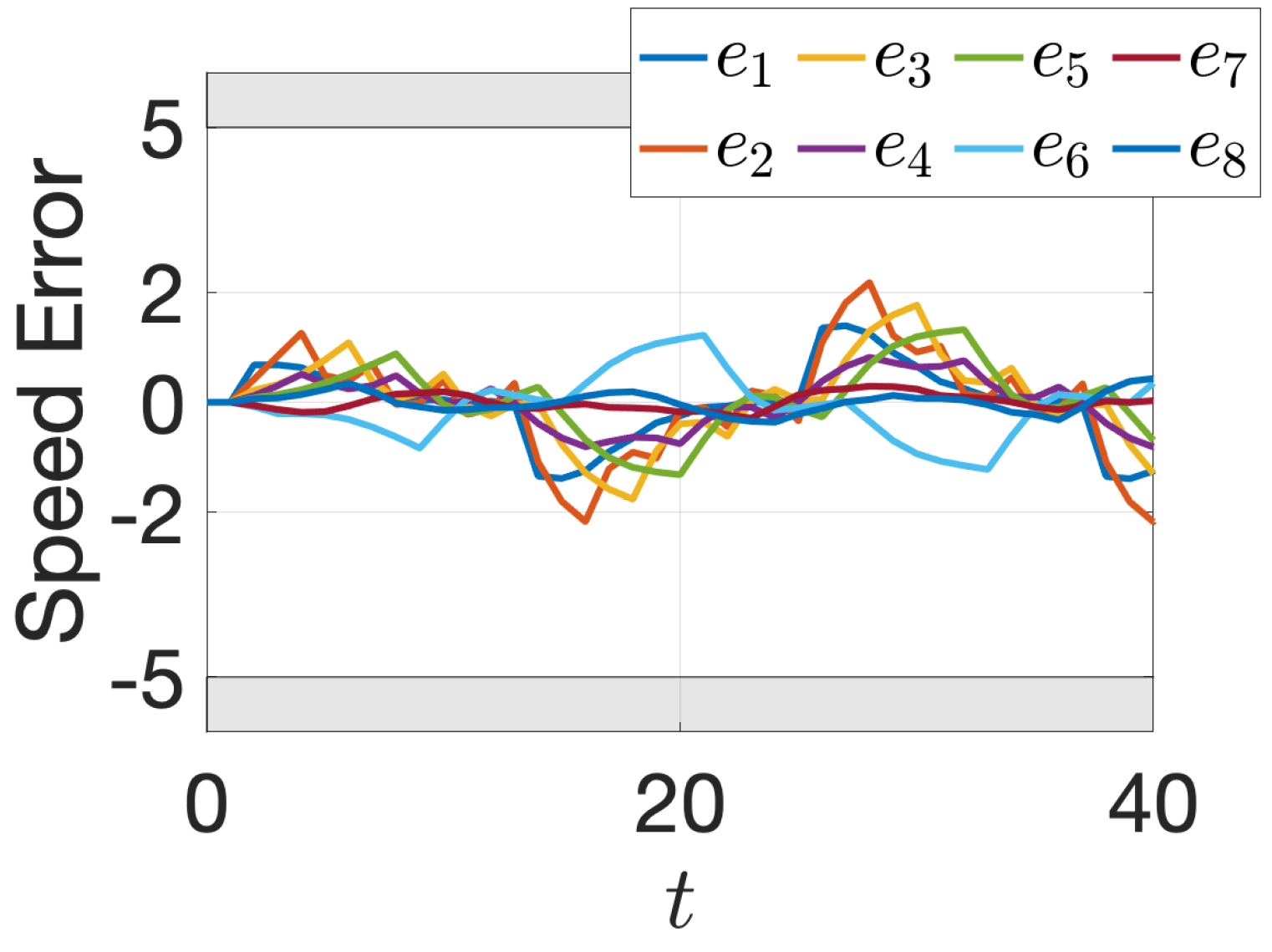}
		\caption{}
	\end{subfigure}
	\begin{subfigure}[b]{0.24\textwidth}
		\centering
		\includegraphics[width=\textwidth]{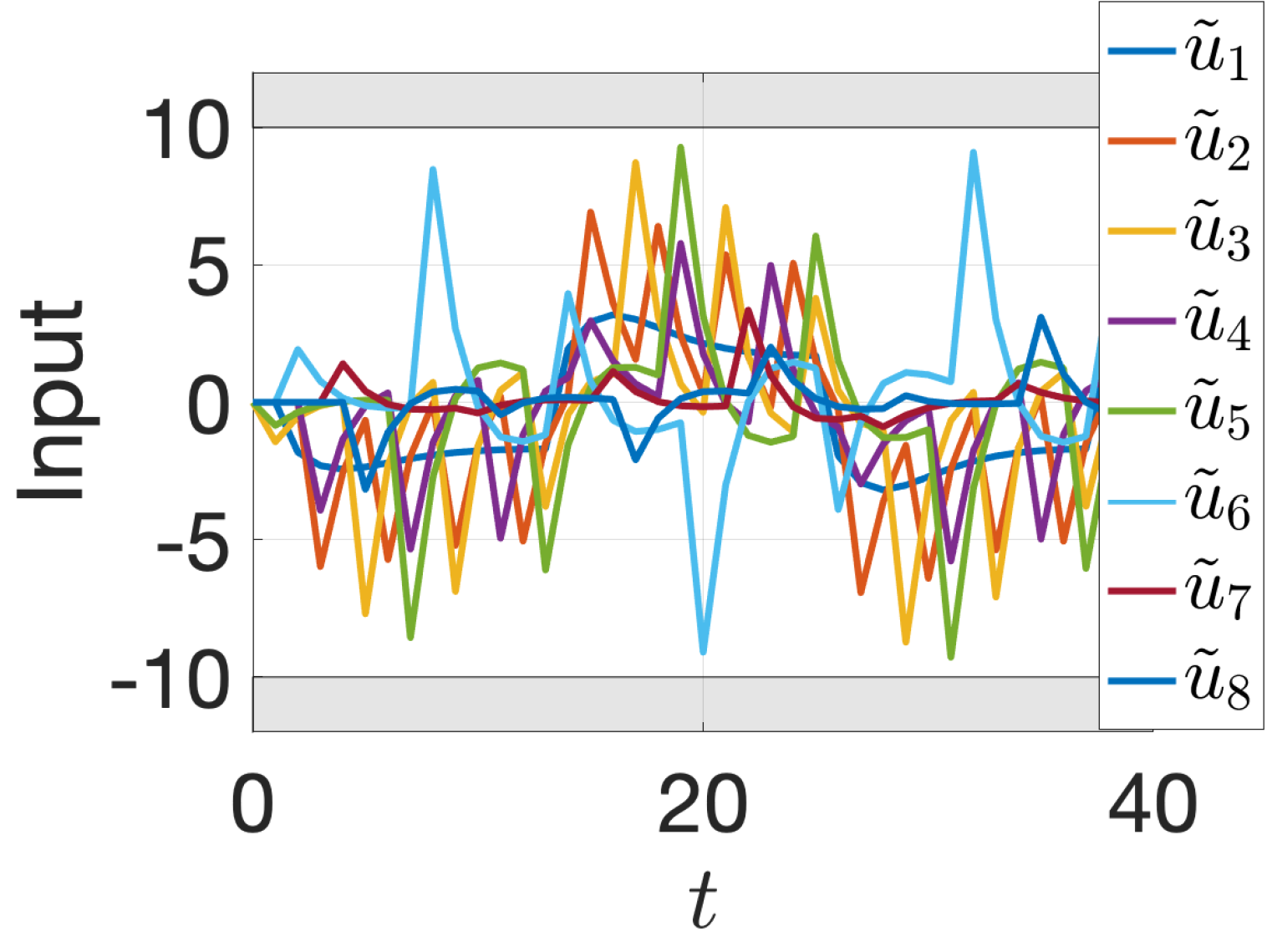}
		\caption{}
	\end{subfigure}
	\hfill
	\begin{subfigure}[b]{0.24\textwidth}
		\centering
		\includegraphics[width=\textwidth]{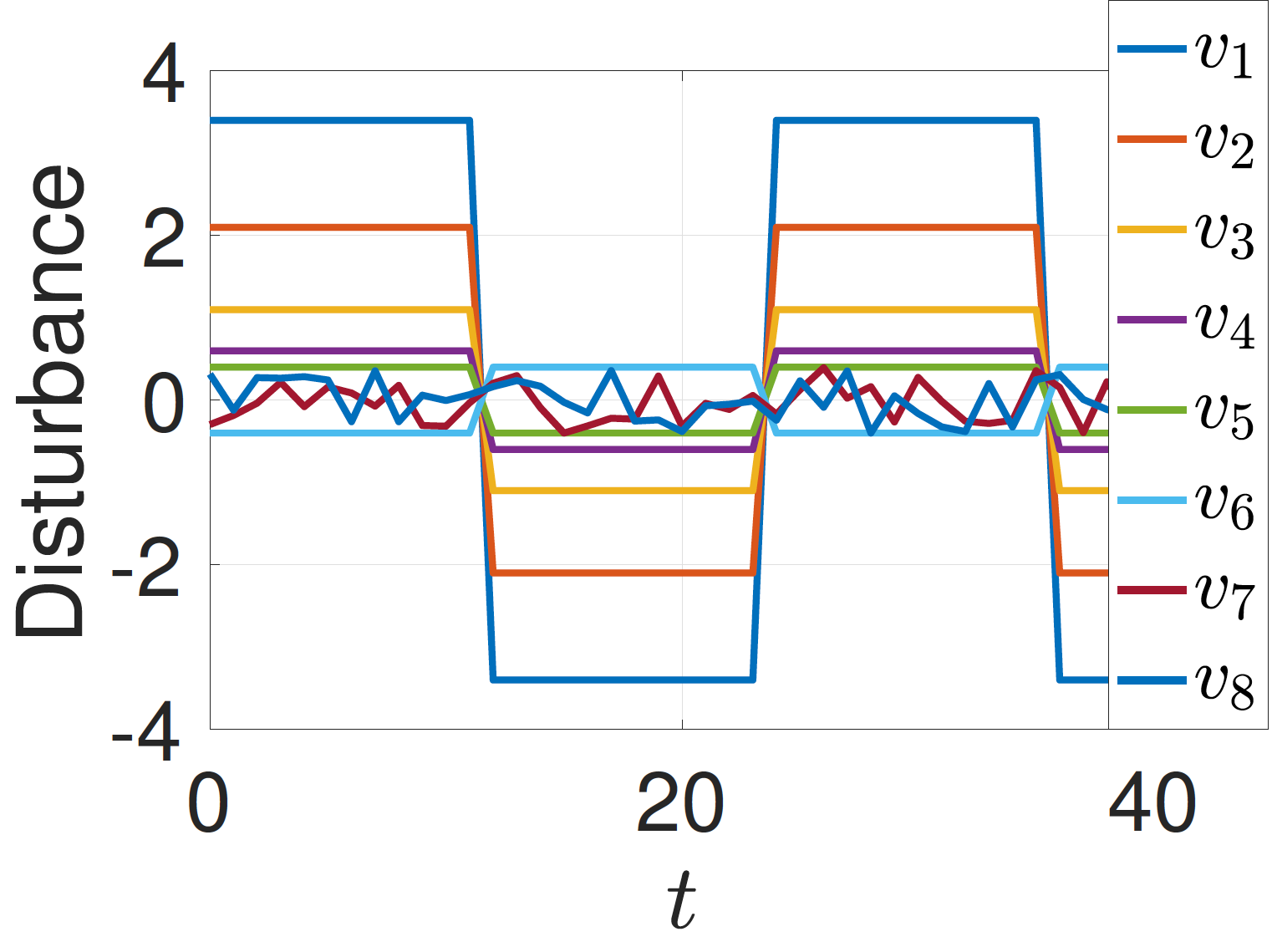}
		\caption{}
	\end{subfigure}
	\caption
	{ Example \ref{ex:ex_tracking}: 
		The shaded bands identify values out of admissible sets
	}
	\label{fig:vehicle_errors}
\end{figure}		
	\end{example}
	
	\begin{example}[Networked control vehicles with packet loss]\label{ex:ex_tracking2}
		Consider the first five vehicles in Example \ref{ex:ex_tracking} in which \(\tilde{v}_1=2,\ \tilde{v}_2=1,\ \tilde{v}_3=0.45,\ \tilde{v}_4=0.25,\ \tilde{v}_5=0.15\). Using \eqref{eq:safe_time}, one can compute safe time intervals as $\alpha_1=4$, $\alpha_2=6$, $\alpha_3=8$, $\alpha_4=10$, and $\alpha_5=12$.
		
		Assuming only two packets can be lost every four successive packets, one can use \eqref{eq:thcon} to calculate new safe time intervals as  $\beta_1=\beta_2=2$, $\beta_3=\beta_4=4$, and $\beta_5=6$.
		
		One can verify that the following $\mathbf{C}_{\mathrm{r}}$ is the cyclic part of a robust feasible schedule.
		\begin{align}
		\mathbf{C}_{\mathrm{r}}=&(1,2),(3,4),(1,2),(1,3),(2,4), \nonumber  \\
		&(1,5),(2,3),(1,4),(2,5)
		\end{align}
		By considering $\mathbf{C}_{\mathrm{r}}$ as the cyclic part of the baseline schedule, one can obtain a robust schedule using either \eqref{eq:schedule_with_losses} or \eqref{eq:Optimal_Online_Schedule_losses}. Note that in this example, we assume that the scheduler has access to all states and there is no measurement noise. The two strategies are compared in Fig.~\ref{fig:vehicle_residual_deadlines}. In this figure, we note:
		\begin{itemize}
			\item The robust safety residuals are non-negative, i.e., $\bar{r}_i(t) \geq 0$, and the update deadlines are positive, i.e., $\gamma_i^x(t) > 0$, both of which imply no constraint is violated;
			\item For the online schedule $\bar{r}_i \geq 3$ most of the times and $\min_i \bar{r}_i =1$ in four time instants; however, in the shifted schedule $\bar{r}_i =1$ at many time instants and  $\min_i \bar{r}_i =0$;
			\item In the online schedule $\max_i \bar{r}_i \leq 6$ at most of the times and $\max_i \bar{r}_i =9$ at just one time instant; however, in the shifted schedule $\max_i \bar{r}_i \leq 8$ most of the times and $\max_i \bar{r}_i =12$ at two time instants. These observations imply that the online schedule increases the safety of the least safe system at the expense of decreasing the safety of safer systems;
			\item One can also see the compromise of the online schedule by comparing the measurement update deadlines. For instance $\min_t \gamma^x_4(t)=2$ and $\min_t \gamma^x_5(t)=4$ for the online schedule, however, $\min_t \gamma^x_4(t)=6$ and $\min_t \gamma^x_5(t)=10$ for the shifted schedule.
		\end{itemize}
		
			\begin{figure}
			\centering
			\begin{subfigure}[b]{0.24\textwidth}
				\centering
				\includegraphics[width=\textwidth]{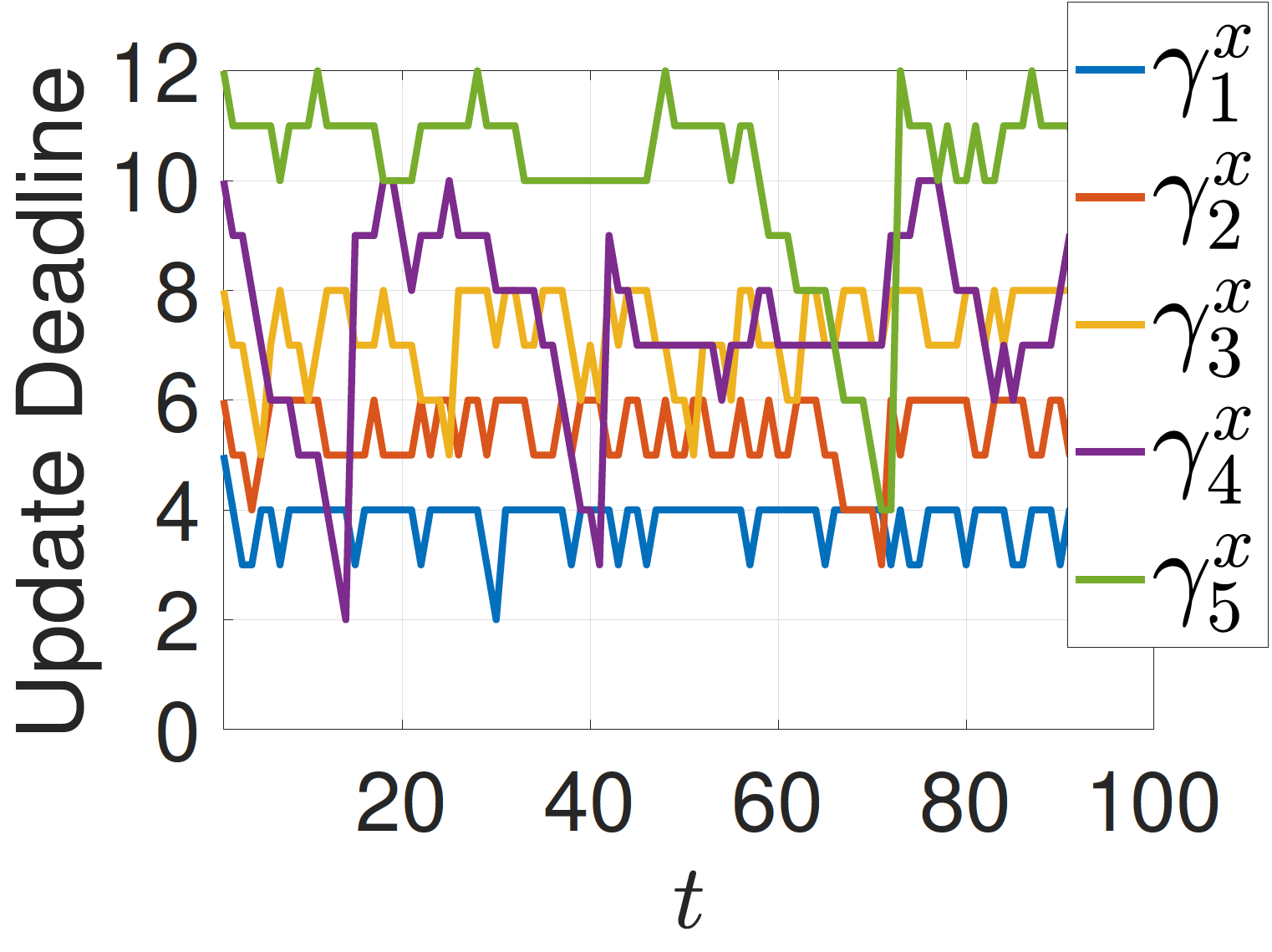}
				\caption{}
			\end{subfigure}
			\hfill
			\begin{subfigure}[b]{0.24\textwidth}
				\centering
				\includegraphics[width=\textwidth]{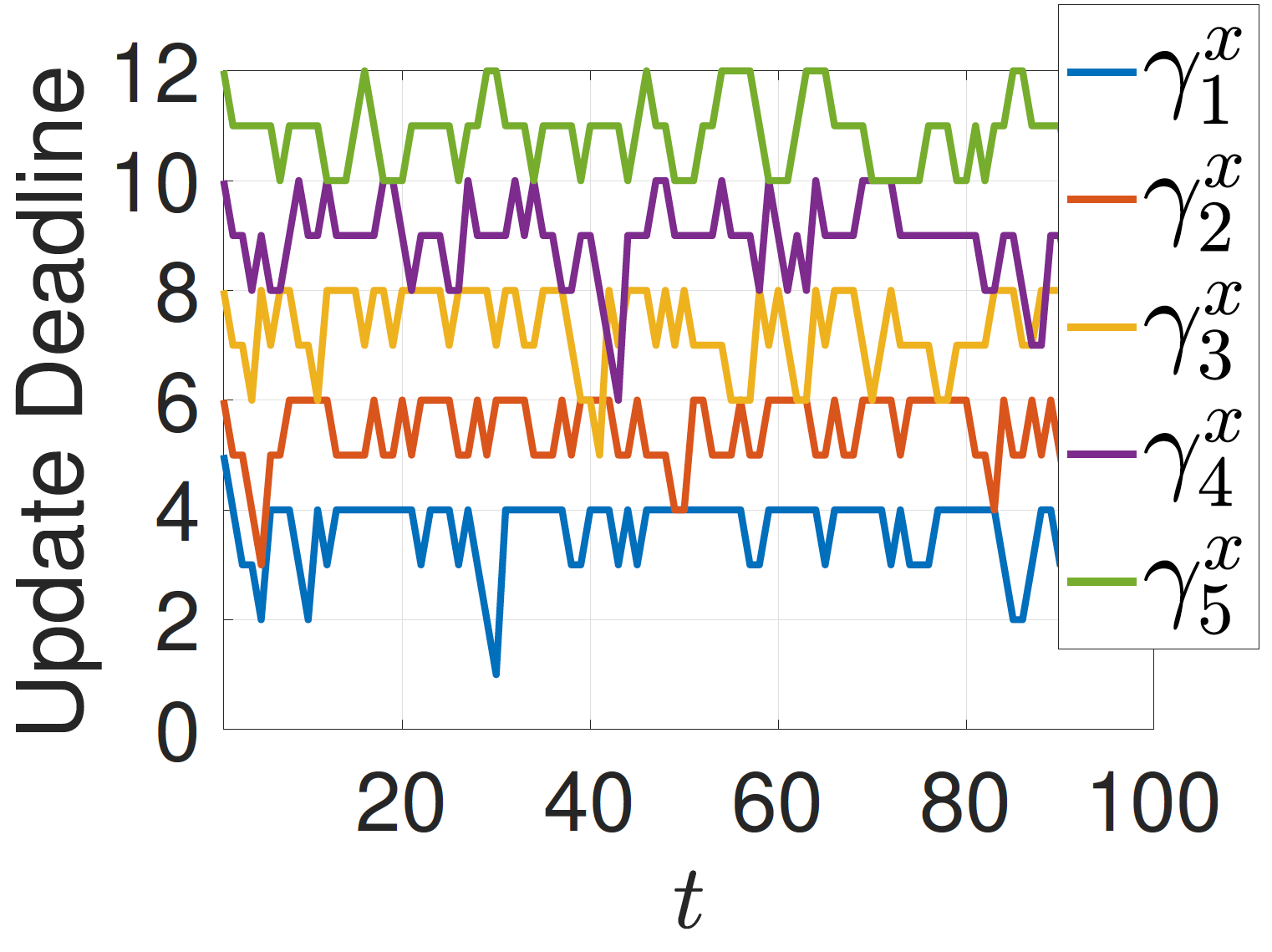}
				\caption{}
			\end{subfigure}
			\begin{subfigure}[b]{0.24\textwidth}
				\centering
				\includegraphics[width=\textwidth]{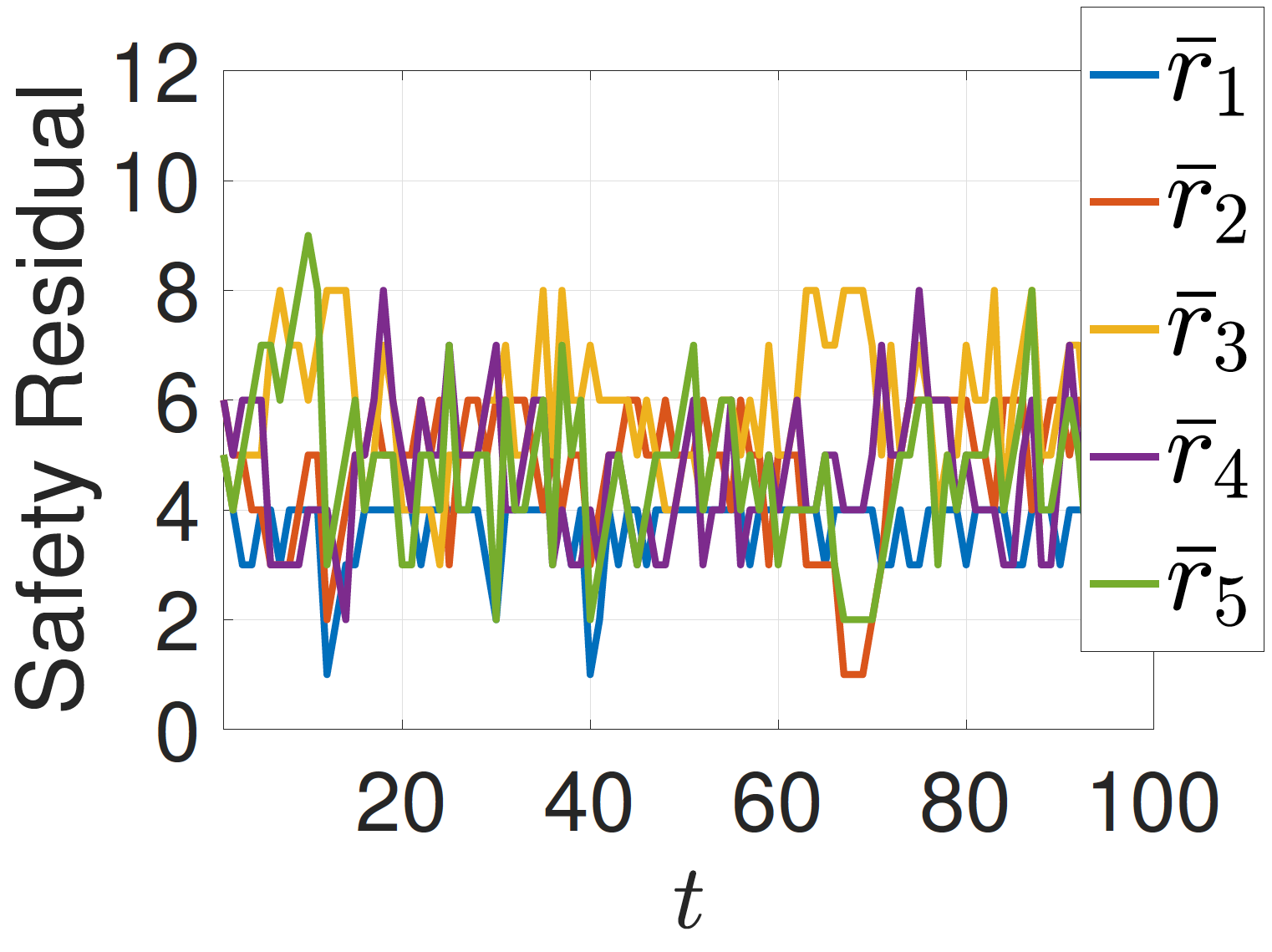}
				\caption{}
			\end{subfigure}
			\hfill
			\begin{subfigure}[b]{0.24\textwidth}
				\centering
				\includegraphics[width=\textwidth]{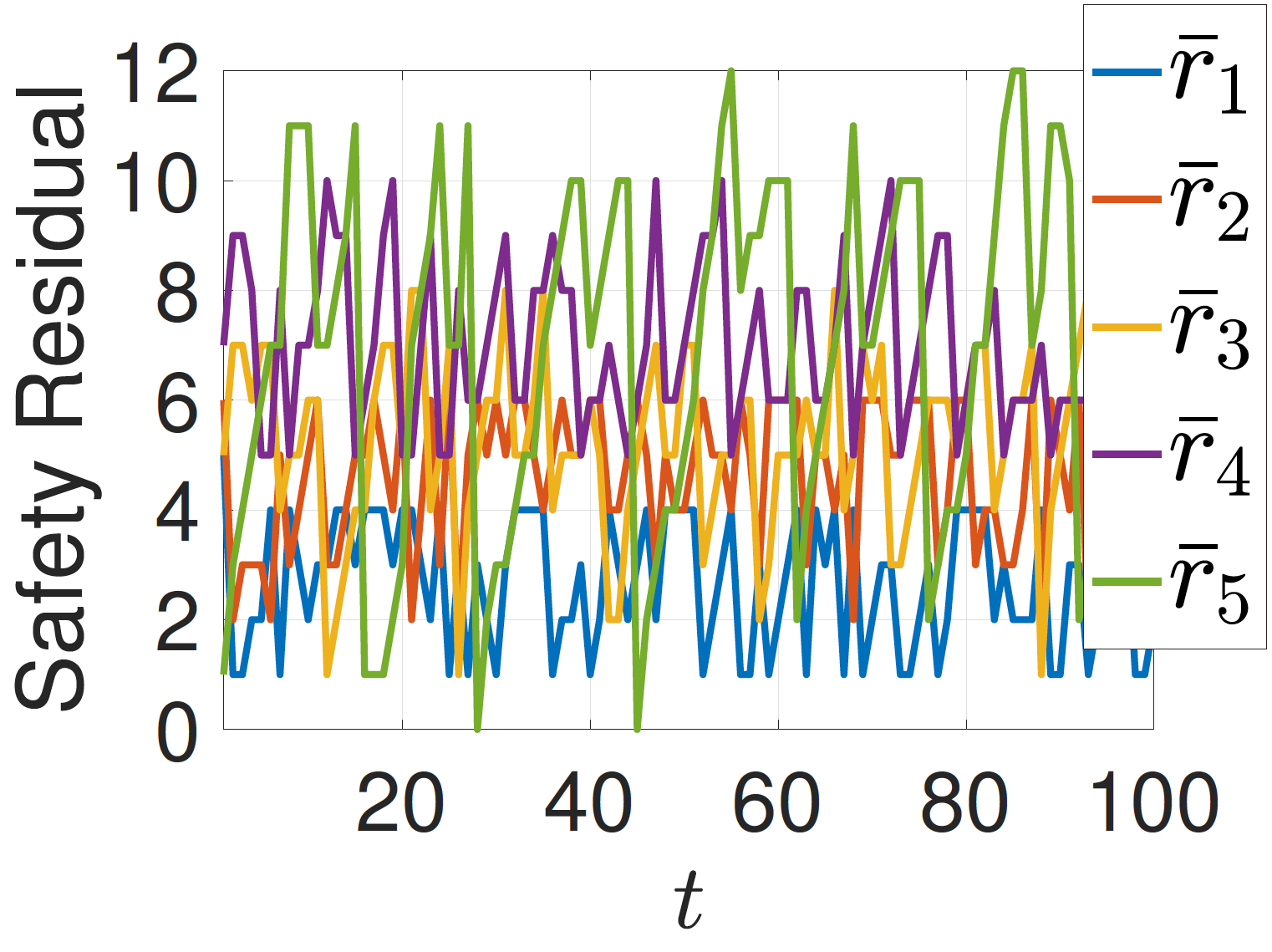}
				\caption{}
			\end{subfigure}
			\caption
			{ Example \ref{ex:ex_tracking2}: 
				Comparison between robust online schedule $\bar{\mathbf{C}}^*$, defined in \eqref{eq:Optimal_Online_Schedule_losses}, and the shifted schedule $\bar{\mathbf{C}}$, defined in \eqref{eq:schedule_with_losses}, that are given in the left and the right columns, respectively. The first row is the measurement update deadlines and the second one is the robust safety residuals defined in \eqref{eq:residuals2}.
			}
			\label{fig:vehicle_residual_deadlines}
		\end{figure}		
		
	\end{example}
	\section{Conclusions}
	\label{section:conclusion}
	In this paper we proposed strategies to guarantee that networked control systems are kept withing an assigned admissible set. We provide such guarantees by translating the control problem into a scheduling problem. To that end, we introduced PP and WSP, reviewed the state-of-the-art knowledge and refined some results on their schedulability. This allowed us to design offline schedules, i.e., schedules which can be applied to NCS regardless of the actual noise realization. In order to reduce conservatism, we proposed an online scheduling strategy which is based on a suitable shift of a pre-computed offline schedule. This allowed us to reduce conservatism while preserving robust positive invariance.
	
	
	Future research directions include designing control laws that maximize the safe time intervals; adopting a probabilistic packet losses model instead of the deterministic one; and considering systems with coupled dynamics or admissible sets.
	
	\bibliographystyle{IEEEtran}
	
	\bibliography{bibfile}

	\begin{IEEEbiography}[{\includegraphics[width=1in,height=1.25in,clip,keepaspectratio]{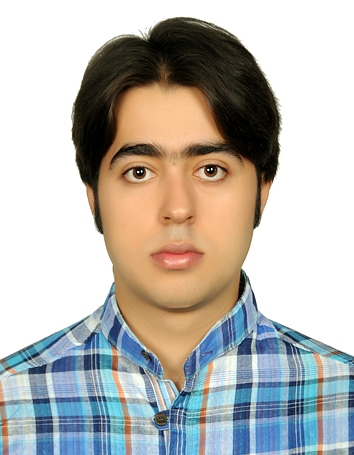}}]{Masoud Bahraini} received his B.Sc. in 2013 form Iran University of Science and Technology, Iran, and his M.Sc. in 2016 form University of Tehran, Iran. He is a Ph.D. student in Chalmers University of Technology since 2017. His research interests are networked control systems, constrained optimal control applied to autonomous and semi- autonomous mobile systems, and control design for nonlinear systems. 
	\end{IEEEbiography}
	\begin{IEEEbiography}[{\includegraphics[width=1in,height=1.25in,clip,keepaspectratio]{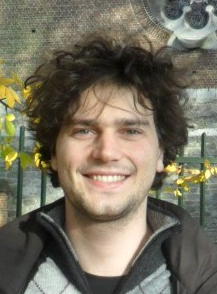}}]{Mario Zanon}
		received the Master's degree in Mechatronics from the University of Trento, and the Dipl\^{o}me D'Ing\'{e}nieur from the Ecole Centrale Paris, in 2010. After research stays at the KU Leuven, University of Bayreuth, Chalmers University, and the University of Freiburg he received the Ph.D. degree in Electrical Engineering from the KU Leuven in November 2015. He held a Post-Doc researcher position at Chalmers University until the end of 2017 and is now Assistant Professor at the IMT School for Advanced Studies Lucca. His research interests include numerical methods for optimization, economic MPC, reinforcement learning, and the optimal control and estimation of nonlinear dynamic systems, in particular for aerospace and automotive applications.
	\end{IEEEbiography}
	\begin{IEEEbiography}[{\includegraphics[width=1in,height=1.25in,clip,keepaspectratio]{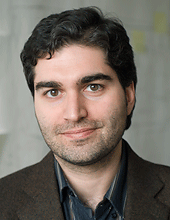}}]{Paolo Falcone}
		received his M.Sc. (Laurea degree) in 2003 from the University of Naples Federico II and his Ph.D. degree in Information Technology in 2007 from the University of Sannio, in Benevento, Italy. He is Professor at the Department of Electrical Engineering of the Chalmers University of Technology, Sweden and Associate Professor at the Engineering Department of Universit\`a di Modena e Reggio Emilia. His research focuses on constrained optimal control applied to autonomous and semi- autonomous mobile systems, cooperative driving and intelligent vehicles, in cooperation with the Swedish automotive industry, with a focus on autonomous
		driving, cooperative driving and vehicle dynamics control.
	\end{IEEEbiography}
	\begin{IEEEbiography}[{\includegraphics[width=1in,height=1.25in,clip,keepaspectratio]{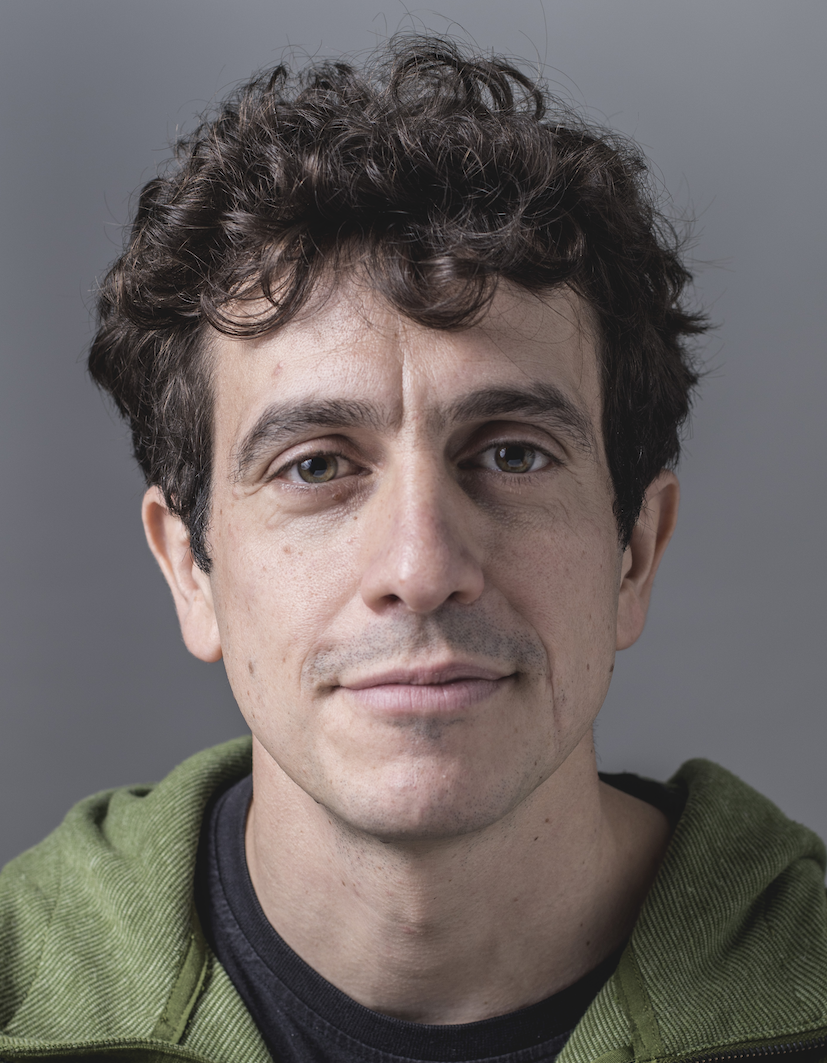}}]{Alessandro Colombo} received the Dipl\^ome D'Ing\'enieur from ENSTA in Paris in 2005, and the Ph.D. from Politecnico di Milano in 2009. He was Postdoctoral Associate at the Massachusetts Institute of Technology in 2010-2012, and is currently Associate Professor in the Department of Electronics, Information and Bioengineering at Politecnico di Milano. His research interests are in the analysis and control of discontinuous, hybrid systems, and networked systems.
	\end{IEEEbiography}
	
\end{document}